\documentclass[acmsmall]{acmart}

\AtBeginDocument{%
  }

\setcopyright{acmcopyright}
\copyrightyear{2023}
\acmYear{2023}
\acmDOI{XXXXXXX.XXXXXXX}

\acmJournal{TMIS}

\usepackage{algorithm}
\usepackage{algorithmic}
\usepackage{bm}
\usepackage{ bbold }
\usepackage{multirow}
\usepackage{caption}
\usepackage{subfigure}
\usepackage{booktabs}
\usepackage{float}
\usepackage{booktabs}
\usepackage{enumitem}

\usepackage{color}
\usepackage{amsmath}
\usepackage{ dsfont }
\usepackage{ upgreek }
\usepackage{booktabs}
\usepackage{multirow}
\usepackage{xcolor}
\usepackage{colortbl}
\usepackage{subfigure}
\usepackage{threeparttable}

\newtheorem{example}{Example}
\newtheorem{definition}{Definition}
\newtheorem{theorem}{Theorem}

\begin{document}

\title{Co-occurrence order-preserving pattern mining}


\author{Youxi Wu}
\affiliation{%
  \institution{School of Artificial Intelligence, Hebei University of Technology}
  \city{Tianjin}
  \country{China}
  \postcode{300401}
}
\email{wuc567@163.com}

	\author{Zhen Wang}
	\affiliation{%
		\institution{School of Artificial Intelligence, Hebei University of Technology}
		\city{Tianjin}
		\country{China}
		\postcode{300401}
	}
 	\affiliation{%
		\institution{ ShiJiaZhuang Posts and Telecommunications Technical College}
		\city{ShiJiaZhuang}
		\country{China}
		\postcode{050021}
	}
 \email{17302297683@163.com}  
 
\author{Yan Li}
\authornote{Corresponding authors.}
\affiliation{%
	\institution{School of Economics and Management, Hebei University of Technology}
	\city{Tianjin}
	\country{China}
	\postcode{300401}
}
\email{lywuc@163.com}

\author{Yingchun Guo}
\affiliation{%
	\institution{School of Artificial Intelligence, Hebei University of Technology}
	\city{Tianjin}
	\country{China}
	\postcode{300401}
}
\email{gyc@scse.hebut.edu.cn}

	\author{He Jiang}
\affiliation{%
	\institution{School of Software, Dalian University of Technology}
	\city{Dalian}
	\country{China}
	\postcode{116023}
}
\email{jianghe@dlut.edu.cn}

\author{Xingquan Zhu}
\affiliation{%
	\institution{Department of Computer \& Electrical Engineering and Computer Science, Florida Atlantic University}
	\city{FL}
	\country{USA}
	\postcode{33431}
}
\email{xzhu3@fau.edu}

\author{Xindong Wu}
	\affiliation{%
\institution{Key Laboratory of Knowledge Engineering with Big Data (the Ministry of Education of China), Hefei University of Technology}
\city{Hefei}
\country{China}
\postcode{230009} 
}
 \email{xwu@hfut.edu.cn}

\renewcommand{\shortauthors}{Y. Wu et al.}

\begin{abstract}
  Recently, order-preserving pattern (OPP) mining has been proposed to discover some patterns, which can be seen as trend changes in time series. Although existing OPP mining algorithms have achieved satisfactory performance, they discover all frequent patterns. However, in some cases, users focus on a particular trend and its associated trends. To efficiently discover trend information related to a specific prefix pattern, this paper addresses the issue of co-occurrence OPP mining (COP) and proposes an algorithm named COP-Miner to discover COPs from historical time series. COP-Miner consists of three parts: extracting keypoints, preparation stage, and iteratively calculating supports and mining frequent COPs. Extracting keypoints is used to obtain local extreme points of patterns and time series. The preparation stage is designed to prepare for the first round of mining, which contains four steps: obtaining the suffix OPP of the keypoint sub-time series, calculating the occurrences of the suffix OPP, verifying the occurrences of the keypoint sub-time series, and calculating the occurrences of all fusion patterns of the keypoint sub-time series. To further improve the efficiency of support calculation, we propose a support calculation method with an ending strategy that uses the occurrences of prefix and suffix patterns to calculate the occurrences of superpatterns. Experimental results indicate that COP-Miner outperforms the other competing algorithms in running time and scalability. Moreover, COPs with keypoint alignment yield better prediction performance. 
\end{abstract}




\maketitle

\section{Introduction}
A time series is a set of numerical values derived from consecutive records over a fixed time interval. Time series exist in many fields, such as market prediction \cite{marketprediction},  taxi demand prediction \cite{taxiprediction}, stock price prediction \cite{stockpriceprediction, temperatureprediction}, and brain EEG analysis \cite{braineeg}. Mining hidden information in time series can be used for big data analysis techniques \cite{wutmis2022}, such as clustering \cite{Hanclustering2022}, classification \cite{deepforest}, pattern recognition \cite{Lirecognition2011, Yanrecognition2020} and predication \cite{Leeprediction2022}. Therefore, time series mining methods are significant and various algorithms have been proposed \cite{Gharghabi1timeseries2020, Gharghabi2timeseries2018}. For example, inspired by the three-way decisions theory \cite{Zhan2threeway2022,seqthreeway}, the NTP-Miner algorithm converts time series to characters and divides the alphabet into three regions: strong, medium, and weak, which can be applied to cluster analysis \cite{WuNTP-Miner2022}. The TBOPE classification algorithm builds an ensemble classifier that can improve classification performance \cite{Baiclassification2021}. The OWSP-Miner algorithm can recognize self-adaptive one-off weak-gap strong patterns, which are more valuable in real life \cite{WuOWSP-Miner2022}. Most time series mining must use representation methods to convert time series into discretized characters \cite{pyramidpattern}. Commonly used methods are piecewise linear approximation (PLA) \cite{StoracePiecewise2004}, piecewise aggregate approximation (PAA) \cite{ChakrabartiLocally2002}, and symbolic aggregate approximation (SAX) \cite{LinExperiencing2007}. However, the above-mentioned symbolization methods easily lead to the loss of trend information inside the original time series. 

Recently, a method without symbolization named order-preserving pattern (OPP) matching has been proposed which can find all occurrences of a given OPP on numeric strings \cite{Kimmatching2014}. An OPP can represent a trend, which is a relative order that can effectively reflect fluctuations without discretizing the values into symbols \cite {copp}. Inspired by OPP matching, a new sequential pattern mining method, named OPP mining was proposed \cite{WuOPP-Miner2022}, which can discover frequent OPPs in time series. To improve the mining performance and discover all strong rules,  order-preserving rule mining was explored \cite{WuOPR-Miner2022}. However, there are two main issues in the current research. 

(1) The existing algorithms cannot cope with distortion and mine redundant patterns, since these algorithms mine directly on the original time series. For example, an OPP is (1,2,3,4,5), which represents an uptrend and can be replaced by the already discovered pattern (1,2), since (1,2) also represents an uptrend. Therefore, it is necessary to discover fewer and more targeted patterns that are more meaningful and unaffected by distortion.

(2) More importantly, the existing algorithms, such as OPP-Miner \cite{WuOPP-Miner2022} and OPR-Miner \cite{WuOPR-Miner2022}, mine all of the patterns and use them for time series clustering and classification. But, in some cases, researchers want to discover patterns related to a specific pattern in historical data, and then utilize a summary of past trends to predict future changes. In this case, mining all the patterns would yield numerous meaningless results. 

To tackle the above two issues, two methods need to be employed, respectively. For distortion, the keypoint alignment method \cite{Lahrechekeypoints2021} can be used.  Thus, to reduce distortion interference in time series, we do not mine on the original time series, but rather on the time series after extracting keypoints. Moreover, to avoid mining redundant patterns, co-occurrence pattern mining can be selected \cite {danguoco, mcor}. The reason is as follows.  Although top-\textit{k} pattern mining \cite{Dong2Top-k2019, Gao3Top-k2016}, closed pattern mining  \cite{Wu1closed2020, Truong2closed2019}, and maximal pattern mining \cite{Li1maximal2022, Vo2maximal2017} can reduce the number of patterns, users may be interested in a part of patterns, rather than all patterns, since all patterns are not sufficiently targeted. For example, when users know a specific prefix pattern in advance to make predictions, they only want to discover patterns with the same specific prefix pattern, and this mining method is called co-occurrence pattern mining \cite {danguoco, mcor}. This problem still exists in time series. For example, in time series prediction, users hope to use historical data to predict future values. Therefore, users know the latest values that can be further extracted as a prefix pattern. Based on the prefix pattern, users can predict future values. Thus, co-occurrence pattern mining is more targeted.   An illustrative example is as follows.


\begin{example}
	Suppose we have a time series \textbf{t} = (16,8,11,10,12,16,17,13,20,18,21,22,18,14,21,24,23,27,25) shown in Fig.~\ref{timeseries}  and a predefined support threshold \textit{minsup}=3. According to OPP-Miner, there are seven frequent OPPs: $\{$(1,2), (2,1), (1,2,3), (1,3,2), (2,1,3), (1,3,2,4), (2,1,3,4)$\}$. Obviously, it is difficult for users to apply these trends information directly. We know that the latest  three values of \textbf{t} are ($t_{17},t_{18},t_{19}$) = (23,27,25) whose relative order is (1,3,2), since $t_{17}$=23 is the smallest value, $t_{18}$=27 is the third smallest value, and $t_{19}$=25 is the second smallest value. Moreover, it is easy to know OPP (1,3,2) occurs four times in \textbf{t}: ($t_2,t_3,t_4$),  ($t_8,t_9,t_{10}$), ($t_{15},t_{16},t_{17}$), and ($t_{17},t_{18},t_{19}$). Thus, (1,3,2) is a frequent OPP. Hence, when we make a prediction, actually, we know the specific prefix pattern.

	\begin{figure}[htbp]
		\centering
		\includegraphics[width=0.75\linewidth]{"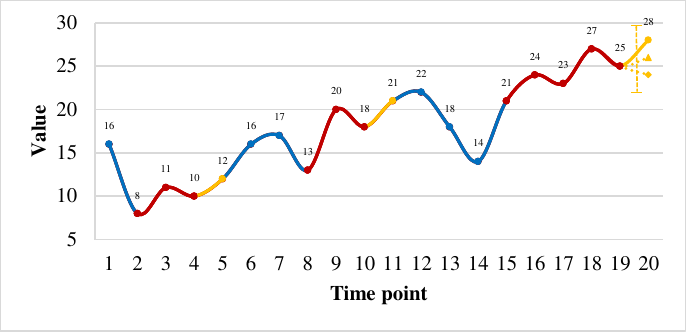"}
		\caption{Time series \textbf{t}. The relative order of sub-time series ($t_2,t_3,t_4$) = (8,11,10) is (1,3,2), since $t_2$=8 is the smallest value, $t_3$=11 is the third smallest value, and $t_4$=10 is the second smallest value. Similarly, the relative orders of  sub-time series ($t_8,t_9,t_{10}$),  ($t_{15},t_{16},t_{17}$), and ($t_{17},t_{18},t_{19}$) are also (1,3,2). Assuming that we know the values $t_{1}$ to $t_{19}$, we can discover the potential trends, and apply these trends and the recent values, such as ($t_{17},t_{18},t_{19}$) to predict the range of $t_{20}$.}
		\label{timeseries}
	\end{figure}

 Now, we use the specific prefix pattern (1,3,2) to predict the further trend. Therefore, we only need to discover the patterns whose prefix pattern is (1,3,2). These discovered patterns are called co-occurrence order-preserving patterns (COPs).  According to historical information, we know that (1,3,2,4) is also a frequent OPP and its prefix pattern is (1,3,2). Hence, (1,3,2,4) is a COP. Thus,  compared with seven frequent OPPs, COP (1,3,2,4) is more targeted in predicting future trends. For example, $t_{20}$ is 28, which is greater than $t_{18}$. Thus, the relative order of ($t_{17},t_{18},t_{19}$, $t_{20}$) is (1,3,2,4), which means that the prediction is correct. If $t_{20}$ is less than $t_{18}$, then the prediction is not correct. 
 
 \label{Example 1}
\end{example}

From the above description, mining COPs with keypoint alignment for time series prediction is meaningful. The main contributions are as follows.

\begin{itemize}
		
	\item To avoid mining irrelevant trends and to obtain better prediction performance for time series, we explore COP mining, which can mine all COPs with the same prefix pattern, and we propose the COP-Miner algorithm.
	
	\item COP-Miner is composed of three parts: extracting keypoints to reduce distortion interference and avoid mining redundant patterns, preparation stage to prepare for the first round of mining, and iteratively calculating supports of superpatterns and mining frequent COPs.

\item We propose the concept of fusion pattern which can effectively reduce the number of candidate patterns.

\item To further improve the efficiency of support calculation, we propose a support calculation method with an ending strategy which uses the occurrences of prefix and suffix patterns to calculate the occurrences of superpatterns. 
 
	\item To validate the performance of COP-Miner, we select ten competing algorithms and nine real time series datasets. Our experimental results verify that COP-Miner outperforms other competing algorithms and COPs with keypoint alignment yield better prediction performance.
	
\end{itemize}

This paper is organized as follows. Section \ref{section2} gives an overview of related work. Section \ref{section3} defines the problem. Section \ref{section4} proposes the COP-Miner algorithm. Section \ref{section5} reports the performance of COP-Miner. Section \ref{section6} concludes this paper.
	
\section{Related work} \label {section2}
Sequential pattern mining (SPM) is an important research direction in big data analysis \cite{patternmining, wangapin, wu2014apin, Mavroudopoulos2022}, which is widely used in bioinformatics \cite{Zhangbioinformatics2021}, keyword extraction \cite{keywu2017, oneoffguo}, feature selection \cite{Wu1Top-k2021,  gan2022, featureselection}, and sequence classification \cite {zengyouhe, li2012tkdd, philippe}. To meet different needs in real life, various pattern mining methods have been proposed \cite {gantkdd}, such as episode pattern mining \cite{Aoepisode2019}, negative SPM \cite{wunegativetkdd}, three-way SPM \cite{Min1threeway2020, threeway}, contrast SPM \cite{copp, Licontrast2020}, high-utility SPM \cite{gengmengkbs, Nawtmis, ganutil}, spatial co-location pattern mining \cite{Wang1colocation2022, Wang2colocation2018}, and gap constraints SPM \cite {nosep}. However, most of these schemes aim to discover all patterns, which leads to numerous information that is not of interest to users. To reduce the number of mined patterns, top-\textit{k} SPM \cite{Dong2Top-k2019, Gao3Top-k2016}, closed SPM  \cite{Wu1closed2020, Truong2closed2019}, and maximal SPM \cite{Li1maximal2022, Vo2maximal2017} have been proposed. However, in some cases, users have certain prior knowledge and want to obtain patterns with a specific prefix pattern. To solve this problem, some co-occurrence SPM methods were proposed, such as maximal co-occurrence SPM\cite{mcor} and top-\textit{k} co-occurrence patterns mining \cite{Amagata1co-occurrence2019}. 

However, the above methods mainly focused on mining character sequences (such as biological and customer purchase behavior sequences), which are difficult to apply directly to time series (such as temperature changes and stock trading volumes). At present, the prevailing solution is to discretize the values into symbols by symbolic methods, such as PLA \cite{StoracePiecewise2004}, PAA \cite{ChakrabartiLocally2002}, and SAX \cite{LinExperiencing2007}, and then apply the above SPM algorithms to find frequent patterns in time series. But the drawbacks of symbolization cannot be ignored. First of all, the symbolization of time series may bring some noise. More importantly, it is difficult to discover the potential trends in time series, since the symbolization methods convert the real values into different symbols which means that these methods focus too much on specific values and ignore trends.

We have noticed a new approach in the field of pattern matching named OPP matching, which does not require symbolization of time series. For example, Kim et al. \cite{Kimmatching2014} proposed an OPP matching solution and indicated trends based on the relative order of numerical values. Moreover, Kim et al. \cite{oppmatchingscaling} studied an approximate order-preserving pattern matching problem with scaling.  
Inspired by OPP matching, our previous work proposed the OPP-Miner algorithm to discover all frequent OPPs \cite{WuOPP-Miner2022} and further proposed the OPR-Miner algorithm to mine strong order-preserving rules from all frequent OPPs \cite{WuOPR-Miner2022}. The common characteristic of OPP-Miner and OPR-Miner is the discovery of all frequent patterns.

To mine patterns with user-specified prefix pattern, inspired by co-occurrence pattern mining \cite{danguoco, mcor}, we present a COP mining scheme for time series prediction. COP mining can discover fewer and more targeted patterns.
Although we can employ OPP-Miner \cite{WuOPP-Miner2022} or EFO-Miner \cite{WuOPR-Miner2022} to discover all frequent OPPs at first, and then select the COPs, it is obviously inefficient and belongs to the brute-force approach, since this approach will discover many irrelevant patterns. Therefore, to discover COPs efficiently, we specially design an algorithm named COP-Miner.

\section{Problem definition}  \label {section3}
\textcolor {black}{This section presents the related concepts of COPs with specific prefix patterns for time series. }

\begin{definition}
\rm{\textit{\textbf{(Original time series)}}} An original time series is a finite sequence of numerical values derived from consecutive records over a fixed time interval, denoted by \textbf{t} = ($t_1,…, t_i,…, t_N$), where 1 $\leqslant i \leqslant N$.
\label{Dedinition 1}
\end{definition}

Considering that a time series can be disturbed easily by distortion, resulting in warping, scaling, and data perturbation, we extract the keypoints of a time series to eliminate distortion.

\begin{definition}
	\rm{\textit{\textbf{(Keypoint time series)}}} For a time series \textbf{t}, $t_i$ is a keypoint, iff it is the local minimum or local maximum and satisfies the following conditions \cite{Lahrechekeypoints2021}. The new time series \textbf{k}, called  keypoint time series, is  a time series composed of all keypoints in \textbf{t}.
	
	\begin{enumerate}
		\item [(1)] If \textit{i}=1 or \textit{i}=\textit{N}, then $t_{i}$ is the first or last point.
		\item [(2)] If $t_i < t_{i-1}$ and $t_i \leqslant t_{i+1}$, or $t_i \leqslant t_{i-1}$ and $t_i < t_{i+1}$, then $t_i$ is a local minimum point.
		\item [(3)] If $t_i > t_{i-1}$ and $t_i \geqslant t_{i+1}$, or $t_i \geqslant t_{i-1}$ and $t_i > t_{i+1}$, then $t_i$ is a local maximum point.
	\end{enumerate}
	\label{Definition 2}
\end{definition}

\textcolor {black}{An illustrative example of time series and keypoint time series is shown as follows.} 

\begin{example}
	Suppose we have a time series \textbf{t} = (16,8,11,10,12,16,17,13,20,18,21,22,18,14,21,24,23,27, 25,28). $t_1$=16 is a keypoint, since it is the first point. $t_2$=8 is a keypoint, since $t_2$=8 $<$ $t_1$=16 and $t_2$=8 $<$ $t_3$=11. However, $t_5$ is not a keypoint, since  $t_4$=10 $<$ $t_5$=12 $<$ $t_6$=16. After extracting all keypoints, \textbf{k} = (16,8,11,10,17,13,20,18,22,14,24,23,27,25,28).
	\label{Example 2}
\end{example}

\begin{definition}
	\rm{\textit{\textbf{(Sub-time series and keypoint sub-time series)}}} Sub-time series \textbf{p} is also a group of numerical values ($p_1,\cdots, p_j,\cdots, p_M$) (1 $\leqslant j \leqslant M$). After extracting keypoints, \textbf{q} = ($q_1,\cdots, q_j,\cdots, q_m$) (1 $\leqslant j \leqslant m$) is a keypoint sub-time series of \textbf{p}.
	\label{Definition 3}
\end{definition}

\begin{definition}
	\rm{\textit{\textbf{(OPP)}}} The rank of an element $q_j$ in a sub-time series \textbf{q} is its rank order, denoted by $o_\textbf{q}(q_j)$. An order-preserving pattern (OPP) is represented by the relative order of \textbf{q}, described by \textbf{o} = \textit{R}(\textbf{q}) = ($o_\textbf{q}(q_1), o_\textbf{q}(q_2),\cdots, o_\textbf{q}(q_m)$).
	\label{Definition 4}
\end{definition}    

\begin{example}   \label{Example 3}
Suppose we have a sub-time series \textbf{p} = (5,3,7,13,8). After extracting keypoints of \textbf{p}, the new keypoints sub-time series \textbf{q} is (5,3,13,8). Since 5 is the second smallest value in \textbf{q}, its rank order is two, i.e., $o_\textbf{q}$(5)=2. Similarly, $o_\textbf{q}$(3)=1, $o_\textbf{q}$(13)=4, and $o_\textbf{q}$(8)=3. The OPP of \textbf{q} is \textit{R}(\textbf{q}) = (2,1,4,3).
	
\end{example}

\textcolor{black}{Now, we give the definitions of occurrence, support, and frequent OPP.}   

\begin{definition}
\rm{\textit{\textbf{(Occurrence and support)}}} For a keypoint time series \textbf{k} = ($k_1,\cdots, k_i,\cdots, k_n$) and  an OPP \textbf{o} = ($o_1,\cdots, o_j,\cdots, o_m$), $\left\langle\textit{i}\right\rangle$ is an occurrence of  \textbf{o} in \textbf{k}, iff \textit{R}(\textbf{k}’) = \textbf{o}, where \textbf{k}’ is a sub-time series ($k_{i-\textit{m}+1}, k_{i-\textit{m}+2}$, $\cdots$, $k_{i}$) ($1 \leqslant i-\textit{m}+1$ and $i \leqslant n$). The number of occurrences of \textbf{o} in \textbf{k} is called the support, denoted by \textit{sup}(\textbf{o}, \textbf{k}).
	\label{Definition 5}
\end{definition}

\begin{definition}
\rm{\textit{\textbf{(Frequent OPP)}}} If the support of OPP \textbf{o} in \textbf{k} is no less than the user-specified support threshold \textit{minsup} (i.e., $sup(\textbf{o},\textbf{k}) \geqslant minsup$), then \textbf{o} is a frequent OPP.
	\label{Definition 6}
\end{definition}

\begin{example}
	Suppose we have the same keypoint time series \textbf{k} as in Example ~\ref{Example 2} and keypoint sub-time series \textbf{q} = (5,3,13,8), and we set \textit{minsup}=3. According to Definition ~\ref{Definition 4}, \textbf{o} = \textit{R}(\textbf{q}) = (2,1,4,3). The relative order of  sub-time series ($k_3, k_4, k_5, k_6$) = (11,10,17,13) are (2,1,4,3). Similarly, $P(k_5, k_6, k_7, k_8)$ = $P(k_9, k_{10}, k_{11}, k_{12})$ = $P(k_{11}, k_{12}, k_{13}, k_{14})$ = (2,1,4,3). Therefore, $\left\langle\textit{6}\right\rangle$, $\left\langle\textit{8}\right\rangle$, $\left\langle\textit{12}\right\rangle$, and $\left\langle\textit{14}\right\rangle$ are four occurrences of \textbf{o} in \textbf{k}. Thus, \textit{sup}(\textbf{o}, \textbf{k}) = 4 is no less than \textit{minsup}. Hence, \textbf{o} = (2,1,4,3) is a frequent OPP.
	\label{Example 4}
\end{example}

\textcolor{black}{Based on the definition of frequent OPP, we present the definitions of prefix OPP, COP, and our problem.}   

\begin{definition}
	\rm{\textit{\textbf{(Prefix OPP and suffix OPP, order-preserving subpattern and order-preserving superpattern)}}} For an OPP  \textbf{x} = ($x_1, x_2,…, x_m$), pattern \textbf{e} = $P(x_1, x_2,…, x_{m-1})$ is the prefix OPP of \textbf{x}, denoted by \textbf{e} = \textit{{prefixorder}}(\textbf{x}), and pattern \textbf{s} = $P(x_2, x_3,…, x_m)$ is the suffix OPP of \textbf{x}, denoted by \textbf{s} = \textit{{suffixorder}}(\textbf{x}). Moreover, \textbf{e} and \textbf{s} are the order-preserving subpatterns of \textbf{x}, and \textbf{x} is the order-preserving superpattern of \textbf{e} and \textbf{s}.
	\label{Definition 7}
\end{definition}

\begin{definition}
	\rm{\textit{\textbf{(COP)}}}  Suppose we have a frequent OPP \textbf{c} = $(c_1, c_2,…, c_k)$. If the relative order of $(c_1, c_2,…, c_m)$ ($m<k$) is \textbf{o} = $(o_1, o_2,…, o_m)$, then pattern \textbf{c} is  a COP of  \textbf{o}.
	\label{Definition 8}
\end{definition}


\begin{definition}
  \rm{\textit{\textbf{(Problem statement)}}} Suppose we have a time series \textbf{t}, a user-specified minimum support threshold \textit{minsup}, and a sub-time series \textbf{p}. After extracting keypoints, we obtain keypoint time series \textbf{k} and  keypoint sub-time series \textbf{q}. Then, we get OPP \textbf{o}=\textit{R}(\textbf{q}). Our goal is to discover all  frequent COPs of  pattern  \textbf{o} in \textbf{k}.
	\label{Definition 9}
\end{definition}

\begin{example}
	We use the same keypoint time series \textbf{k} as in Example ~\ref{Example 2} and  keypoint sub-time series \textbf{q} = (5,3,13,8), and we set \textit{minsup}=3. According to Exampel ~\ref{Example 4}, \textit{R}(\textbf{q}) = (2,1,4,3). According to Definition ~\ref{Definition 5}, we know that $\left\langle\textit{7}\right\rangle$, $\left\langle\textit{9}\right\rangle$, $\left\langle\textit{13}\right\rangle$, and $\left\langle\textit{15}\right\rangle$ are four occurrences of pattern \textbf{c} = (2,1,4,3,5). Thus, pattern \textbf{c} is a frequent COP of \textbf{q}, since its prefix OPP is  (2,1,4,3) which is the same as \textit{R}(\textbf{q}).
	\label{Example 5}
\end{example}

For clarification, the used symbols are listed in Table ~\ref{tab1}.

\begin{table}[htbp]
	\centering
		\scriptsize
		\caption{Notation description}
		\resizebox{9cm}{!}{
			\begin{tabular}{ll}
				\toprule
				Symbol & Description  \\
				\midrule    
				\textbf{t} = $(t_1,\cdots, t_i,\cdots, t_N)$ & Original time series \\
				\textbf{k} = $(k_1,\cdots, k_i,\cdots, k_n)$ & Keypoint time series \\
				\textbf{p} = $(p_1,\cdots, p_j,\cdots, p_M)$ & User-specified sub-time series \\
				\textbf{q} = $(q_1,\cdots, q_j,\cdots, q_m)$ & Keypoint sub-time series \\
                    \textbf{o} = $(o_1,\cdots, o_j,\cdots, o_m)$ & A frequent OPP, \textbf{o} = \textit{R}(\textbf{q})\\
                    $Occ_\textbf{o}$ = $\{go_1, \cdots, go_i, \cdots, go_y\}$ & Occurrences of OPP \textbf{o} \\
				\textbf{e} = \textit{{prefixorder}}(\textbf{x}) & Prefix OPP of \textbf{x} \\
				\textbf{s} = \textit{{suffixorder}}(\textbf{x}) & Suffix OPP of \textbf{x} \\
				\textbf{b} = $(b_1,\cdots, b_i,\cdots, b_{n-1})$ & Binary sequence after time series \textbf{k} conversion \\
				\textbf{d} = $(d_1,\cdots, d_j,\cdots, d_{m-1})$ & Binary sequence after pattern \textbf{s} conversion \\
				\textbf{u} = $(u_1,\cdots, u_i,\cdots, u_{m+1})$ & \textbf{u} = \textbf{o}$\oplus$\textbf{f}, superpattern with length \textit{m}+1 \\
				\textbf{v} = $(v_1,\cdots, v_i,\cdots, v_{m+1})$ & \textbf{v} = \textbf{o}$\oplus$\textbf{f}, superpattern with length \textit{m}+1 \\
				\textbf{c} = $(c_1, c_2,\cdots, c_k)$ & COP of pattern \textbf{o} \\ 
    
				\bottomrule
		\end{tabular}}
		\label{tab1} 
	\end{table}

\section{Proposed algorithm} \label {section4}

\textcolor {black}{ Although OPP-Miner \cite{WuOPP-Miner2022} and EFO-Miner \cite{WuOPR-Miner2022} can discover all frequent patterns, for COP mining, it is inefficient to obtain all frequent patterns and then to discover COPs.   To tackle COP mining, we propose COP-Miner, which has three parts: a local extremum point (LEP) algorithm to extract keypoints shown in Section \ref{subsect4.1}, a preparation stage to realize the first round of mining shown in Section \ref{subsect4.2}, and iteratively calculating supports and mining frequent COPs shown in Section \ref{subsect4.3}. Fig. ~\ref{FrameWork} shows the framework of COP-Miner and the COP-Miner algorithm and its time and space complexities analysis are shown in Section \ref{subsect4.4}.}

\begin{figure}[htbp]
	\centering
	\includegraphics[width=0.99\linewidth]{"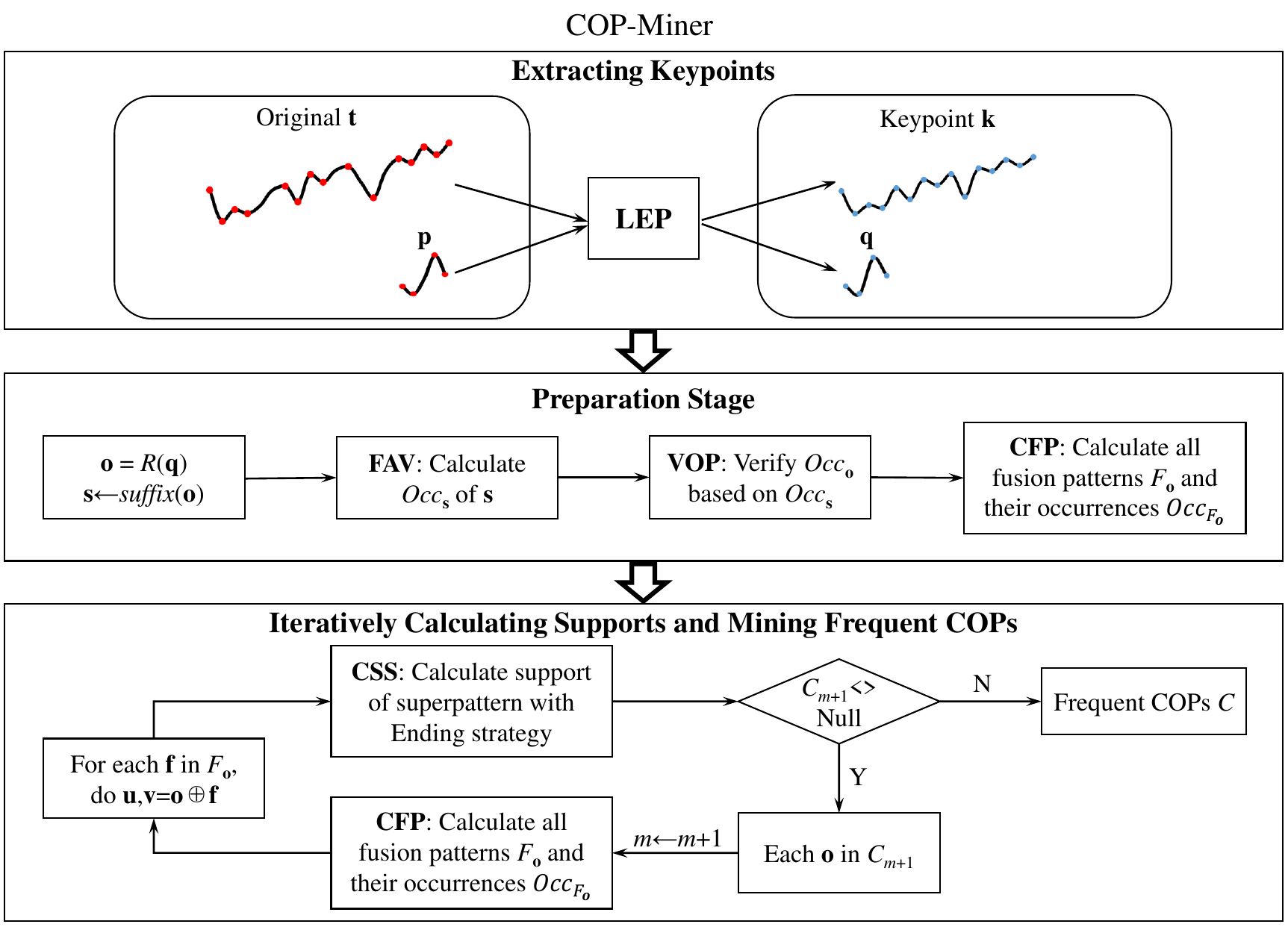"}
	\caption{Framework of COP-Miner. There are three parts: extracting keypoints, preparation stage, and iteratively calculating supports and mining frequent COPs whose key issue is to calculate supports of superpatterns.}
	\label{FrameWork}
\end{figure}

\subsection{Extracting keypoints} \label {subsect4.1}

Traditional OPP mining operates directly on the original time series to find sub-time series with the same trend. But there are two issues that should be considered. On the one hand, there is a lot of redundant information in the mining results. For example, OPP (1,2) indicates an upward trend, and  OPP (1,2,3) also indicates a continuous upward trend. On the other hand, point-by-point analysis of time series cannot cope with sequences in noisy environments like scaling, warping, and shifting. Therefore, the first key task in COP mining is to extract keypoints. Hence, we extract keypoints according to Definition \ref{Definition 2}, and the pseudo-code of the local extremum point (LEP) algorithm is shown in Algorithm ~\ref{Algom 1}.

\begin{algorithm}[htbp]
	\caption{LEP} 
	\label{Algom 1}
		\hspace*{0.02in} \leftline{{\bf Input:} Time series \textbf{t}}\\
		\hspace*{0.02in} \leftline{{\bf Output:} Keypoint time series \textbf{k}}\\		
	\begin{algorithmic}[1]
		\STATE Append $t_1$ to \textbf{k};
		\FOR {$i$= 2 to $N-1$ } 
		\IF {($t_i<t_{i-1} \wedge t_i \leqslant t_{i+1}) \| (t_i \leqslant t_{i-1} \wedge t_i<t_{i+1})$}
		\STATE Append $t_i$ to \textbf{k}; 
		\ELSIF {($t_i>t_{i-1} \wedge t_i \geqslant t_{i+1}) \| (t_i \geqslant t_{i-1} \wedge t_i>t_{i+1})$}
		\STATE Append $t_i$ to \textbf{k}; 
		\ENDIF
		\ENDFOR
            \STATE  Append $t_N$ to \textbf{k};
		\RETURN \textbf{k};
	\end{algorithmic}
\end{algorithm}

\subsection{Preparation stage}\label {subsect4.2}

This section is designed to prepare for the first round of mining. The related definitions and examples are as follows.

\begin{definition}\label{Definition 10}
	\rm{\textit{\textbf{(Fusion pattern)}}} For a keypoint sub-time series \textbf{q} and its corresponding OPP \textbf{o} = \textit{R}(\textbf{q}), if \textit{{suffixorder}}(\textbf{o}) = \textit{{prefixorder}}(\textbf{f}), then \textbf{f} = ($f_1, f_2,…, f_{m-1}, f_m$) is called a fusion pattern of OPP \textbf{o}.
\end{definition}

There are several methods to calculate the occurrences of the prefix OPP and suffix OPP.

\begin{enumerate}
\item [1)] A simple method is to enumerate all possible patterns and to use a pattern matching method to calculate the supports of these patterns. The principle of the enumeration method is as follows. We enumerate all fusion patterns. Suppose the length of \textit{{suffixorder}(\textbf{o})} is \textit{m}-1. It is easy to know there are \textit{m} fusion patterns. Thus, we need to calculate the occurrences of \textit{m}+1 patterns using pattern matching method. Obviously, this is a brute-force method.

\item [2)]  An efficient method is to reduce the number of support calculations using the pattern matching method, since this is the most time-consuming step. To tackle this issue effectively, according to the principle of EFO-Miner \cite{WuOPR-Miner2022}, we need to calculate the occurrences of pattern \textit{{suffixorder}}(\textbf{o}) and its fusion  patterns.

\end{enumerate}

In this paper, we adopt the second method. For example, suppose the user specified a sub-time series (5,3,7,13,8). After extracting keypoints, the corresponding keypoint sub-time series is (5,3,13,8). \textbf{o} = \textit{R}(\textbf{q}) = (2,1,4,3), which indicates that \textbf{o} is the prefix pattern specified by user, and \textbf{o} is an OPP corresponding to keypoint sub-time series. Firstly, we need to obtain \textbf{o}'s suffix OPP (1,3,2), since this is the common part of prefix pattern and its fusion patterns. Secondly, we employ the pattern matching method to calculate the occurrences of pattern (1,3,2). Finally, according to the principle of EFO-Miner \cite{WuOPR-Miner2022}, to calculate the supports of candidate superpatterns, we calculate the occurrences of the prefix pattern, i.e., the occurrences of pattern (2,1,4,3) and the occurrences of the suffix patterns, i.e., the occurrences of all fusion patterns of pattern (2,1,4,3).

Therefore, in the preparation stage, we have four steps: obtaining \textbf{s} = \textit{suffixorder}(\textbf{o}), calculating the occurrences $Occ_\textbf{s}$ of pattern \textbf{s}, verifying the occurrences $Occ_\textbf{o}$ of pattern \textbf{o} based on $Occ_\textbf{s}$, and calculating the occurrences of all fusion patterns of pattern \textbf{o}.		

\textbf{Step 1.}  Obtain suffix OPP \textbf{s} of pattern \textbf{o} according to Definition ~\ref{Definition 7}.

\begin{example}
	Suppose we have a pattern \textbf{o} = (2,1,4,3). Thus, its suffix is (1,4,3), and the corresponding suffix OPP is \textbf{s} = (1,3,2).
	\label{Example 6}
\end{example}

\textbf{Step 2.} Use the filtration and verification (FAV) algorithm \cite{WuOPP-Miner2022} to calculate the occurrences $Occ_\textbf{s}$ of pattern \textbf{s}. The FAV has three steps: filtration, SBNDM$_{2}$ matching, and verification. The filtration strategy is adopted to convert the time series and pattern into a binary sequence and a binary pattern. SBNDM$_{2}$ matching  \cite{DurianImproving2010} is used to find the possible occurrences of a binary pattern in a binary sequence. The verification strategy is employed to check whether the possible occurrences are occurrences. \textcolor {black}{This step is called filtration and verification (FAV).}

Example \ref{FAV} illustrates the principle of FAV.

\begin{example} \label {FAV}
Suppose we have the same keypoint time series \textbf{k} as in Example ~\ref{Example 2} and \textbf{s} = (1,3,2). Table ~\ref{tab2} shows the element index of \textbf{k}. According to the filtration strategy, keypoint time series \textbf{k} is transformed into a binary sequence \textbf{b}. If $k_i<k_{i+1}$, then $b_i$=1; otherwise, $b_i$=0. 
For example, $k_1$ and $k_{2}$ are 16 and 8, respectively. Thus, $k_1 > k_{2}$, i.e., $b_1$=0. Similarly, $b_2$=1. Finally, \textbf{b} = (0,1,0,1,0,1,0,1,0,1,0,1,0,1). Moreover, $s_1 < s_{2}$, since  $s_1$ and $s_2$ are 1 and 3, respectively. Thus, $d_1$=1. Similarly, $d_2$=0. Therefore, \textbf{s} is transformed into \textbf{d} = (1,0). Then, according to SBNDM$_2$ matching, we know that $\left\langle\textit{3}\right\rangle$  is a possible occurrence of pattern \textbf{d} in sequence \textbf{b}. Similarly, $\left\langle\textit{5}\right\rangle$, $\left\langle\textit{7}\right\rangle$, $\left\langle\textit{9}\right\rangle$, $\left\langle\textit{11}\right\rangle$, and $\left\langle\textit{13}\right\rangle$ are also possible occurrences. Finally, according to the verification strategy, we verify whether $\left\langle\textit{3}\right\rangle$ is an occurrence. For pattern \textbf{d} = (1,0), $\left\langle\textit{3}\right\rangle$ in \textbf{b} means ($b_2, b_3$) whose corresponding sub-time series in \textbf{k} is ($k_2, k_3, k_4$). $\left\langle\textit{4}\right\rangle$ is the occurrence of pattern \textbf{s} in \textbf{k}, since $k_2<k_4$, which is consistent with $s_1<s_3$. Similarly, $\left\langle\textit{6}\right\rangle$, $\left\langle\textit{8}\right\rangle$, $\left\langle\textit{12}\right\rangle$, and $\left\langle\textit{14}\right\rangle$ are also occurrences  of pattern \textbf{s}, while $\left\langle\textit{10}\right\rangle$ is not.
	\label{Example 7}
	\end{example}
 
	\begin{table}[!htb]
				\renewcommand{\arraystretch}{1.1}    
				\footnotesize %
				\captionsetup{font=footnotesize} %
				\caption{Element index of \textbf{k}} %
				\centering %
				\label{tab2} %
				\tabcolsep 8pt 
				\begin{tabular}{cccccccccccccccc} 
					\toprule 
					Index & 1 & 2 & 3 & 4 & 5 & 6 & 7 & 8 & 9 & 10 & 11 & 12 & 13 & 14 & 15  \\\hline
					\textbf{k} & 16 & 8 & 11 & 10 & 17 & 13 & 20 & 18 & 22 & 14 & 24 & 23 & 27 & 25 & 28 \\
					\bottomrule
				\end{tabular}
			\end{table}

	\textbf{Step 3.} Verify the occurrences $Occ_\textbf{o}$ of pattern \textbf{o} based on $Occ_\textbf{s}$. We know that \textbf{s} = ($s_1, s_2,\cdots, s_{m-1}$) is the suffix pattern of \textbf{o} = ($o_1, o_2,\cdots, o_m$). Suppose $\left\langle\textit{c}\right\rangle$ is an occurrence of \textbf{s}, which means that the OPP of sub-time series ($k_{c-m+2}, k_{c-m+3}, \cdots, k_c$) is \textit{R}(\textbf{s}). Now, we verify whether $\left\langle\textit{c}\right\rangle$ is also an occurrence of \textbf{o}, which means verifying whether the OPP of sub-time series ($k_{c-m+1}, k_{c-m+2}, k_{c-m+3},\cdots$ $, k_c$) is \textbf{o}  based on \textit{R}(\textbf{s}). To verify it easily, we propose the lower position and upper position at first.

 
	\begin{definition}  \label{Definition 11}
		\rm{\textit{\textbf{(Lower position and upper position)}}} We compare the values in ($s_1, s_2,…, s_{m-1}$) and ($q_2,…, q_{m-1}, q_m$) correspondingly. If $s_j=q_{j+1}$, then \textit{j}+1 is an unchanged position; otherwise, \textit{j}+1 is a changed position. Among the unchanged positions, if the value in $l_p$ is the maximum value, then $l_p$ is called the lower position. Among the changed positions, if the value in $u_p$ is the minimum value, then $u_p$ is called the upper position.
	\end{definition}

\begin{example}
	In Example ~\ref{Example 6}, we know \textbf{s} = (1,3,2) and \textbf{o} = (2,1,4,3). Now, we compare the values in (1,3,2) and (1,4,3) correspondingly. Obviously, 2 is an unchanged position, since $s_1$=$o_2$=1. Moreover, 2 is the lower position, since 2 is the only unchanged position. Similarly, $\{$3,4$\}$ are changed positions, since $s_2$=3 $\neq$ $o_3$=4 and $s_3$=2 $\neq$ $o_4$=3. Among the changed positions $\{$3,4$\}$, the value in 4 is the minimum value, since $o_4$=3 $<$ $o_3$=4. Thus, 4 is the upper position.
	\label{Example 8}
\end{example}

Then, based on lower and upper positions, we verify whether $k_{c-m+1}>k_{c-m+l_{p}}$ and $k_{c-m+1}<k_{c-m+u_p}$ are true. If it is true, then $\left\langle\textit{c}\right\rangle$ is also an occurrence of \textbf{o}; otherwise, $\left\langle\textit{c}\right\rangle$ is not. It is worth noting that $k_{c-m+l_{p}}$ or $k_{c-m+u_p}$ can be null. If it is null, the corresponding comparison is pruned.

\begin{example}
	In Example ~\ref{Example 8}, we know that 2 and 4 are the lower and upper positions, respectively.
Moreover, the length of pattern \textbf{s} = (1,3,2) is 3. Thus, \textit{m}=3+1=4.  For occurrence $\left\langle\textit{4}\right\rangle$ in Example ~\ref{Example 7}, \textit{c}-\textit{m}+1=1. Then, we verify whether $k_1>k_2$ and $k_1<k_4$ are true or not. $\left\langle\textit{4}\right\rangle$ is not an occurrence of pattern \textbf{o} in \textbf{k}, since $k_1>k_4$. Similarly, for occurrence $\left\langle\textit{6}\right\rangle$, we verify whether $k_3>k_4$ and $k_3<k_6$. We know that $\left\langle\textit{6}\right\rangle$ is an occurrence, since $k_3$=11, $k_4$=10, and $k_6$=13. Moreover, $\left\langle\textit{8}\right\rangle$, $\left\langle\textit{12}\right\rangle$, and $\left\langle\textit{14}\right\rangle$ are three occurrences of \textbf{o} in \textbf{k}. Now, we validate that $k_{c-m+l_{p}}$ or $k_{c-m+u_p}$ can be null. Suppose \textbf{o} = (4,1,3,2) and \textbf{s} = (1,3,2). Thus, $\{$2,3,4$\}$ are unchanged positions, while the changed position is null. Therefore, 3 is the lower position and the upper position is null, which means that we only need to verify whether $k_{c-m+1}>k_{c-m+l_{p}}$.
	\label{Example 9}
\end{example}

The pseudo-code of verification occurrences for pattern (VOP) \textbf{o}  is shown in Algorithm ~\ref{Algom 2}.

\begin{algorithm}[htb]
	\caption{VOP} %
	\label{Algom 2}
	\hspace*{0.02in} \leftline{{\bf Input:} OPP \textbf{o} and its suffix pattern \textbf{s}, occurrences of \textbf{s}: $Occ_\textbf{s}$, keypoint time series \textbf{k}}\\
		\hspace*{0.02in} \leftline{{\bf Output:} Occurrences of \textbf{o}: $Occ_\textbf{o}$}\\
		\begin{algorithmic}[1]		
		
		\STATE Calculate $l_p$ and $u_p$ according to Definition ~\ref{Definition 10};
		\STATE \textit{m}←$\left| \textbf{s} \right|$+1; // $\left| \textbf{s} \right|$ is the length of \textbf{s}
		\FOR { each \textit{c} in $Occ_\textbf{s}$}
		\IF {$k_{c-m+1}>k_{c-m+l_{p}}  \wedge k_{c-m+1}<k_{c-m+u_p}$}
		\STATE Append \textit{c} to $Occ_\textbf{o}$;
		\ENDIF
		\ENDFOR
		\RETURN $Occ_\textbf{o}$;
	\end{algorithmic}
\end{algorithm}

\textbf{Step 4.} Calculate the occurrences of all fusion patterns of pattern \textbf{o}. In this step, the main issue is that given pattern \textbf{s}, sub-time series ($k_{c-m+2}, k_{c-m+3},…, k_c$) according to occurrence $\left\langle\textit{c}\right\rangle$, and new value $k_{c+1}$, we calculate the OPP of \textbf{k}’ = ($k_{c-m+2}, k_{c-m+3},…, k_c, k_{c+1}$), i.e., \textit{R}(\textbf{k}’). An easy method is to adopt a linear search method by comparing $k_{c+1}$ with $k_{c-i}$ ($0\leq i\leq m-2$) to obtain \textbf{f} = ($f_1, f_2,…, f_{m-1}, f_m$) based on ($s_1, s_2,…, s_{m-1}$). Obviously, this method is inefficient. An efficient method is to employ the binary search method to calculate the relative order of $k_{c+1}$, i.e., \textit{v} = $o_\textbf{k’}(k_{c+1})$, which means that $k_{c+1}$ is the \textit{v}-th smallest in \textbf{k}’. Furthermore, each $\left\langle\textit{c}\right\rangle$ in $Occ_\textbf{s}$ is calculated according to the binary search method. All fusion patterns $F_{\textbf{o}, v}$ are stored into a set $F_{\textbf{o}}$, and corresponding occurrences $Occ_{F_{\textbf{o}, v}}$ are stored into $Occ_{F_{\textbf{o}}}$.

\begin{example}
	We use Table ~\ref{tab2} and Example ~\ref{Example 7} to illustrate this example. We know that \textbf{s} = (1,3,2) and $Occ_\textbf{s}$ = $\{$4,6,8,12,14$\}$. Taking $\left\langle\textit{4}\right\rangle$ as an example, we calculate the OPP of (8,11,10,17). We can employ the binary search method, since (1,3,2) is the relative order of (8,11,10). It is easy to know that the rank order of 17 in (8,11,10,17) is 4. Thus, $\left\langle\textit{5}\right\rangle$ as an occurrence of (1,3,2,4) is stored into $Occ_{F_{\textbf{o},4}}$. Similarly, we know that the rank orders of $k_7$, $k_9$, $k_{13}$, and $k_{15}$  are 4. Hence, $Occ_{F_{\textbf{o},4}}$ = $\{$5,7,9,13,15$\}$, and $F_{\textbf{o},4}$ is (1,3,2,4).
	\label{Example 10}
\end{example}


The pseudo-code of calculating fusion patterns (CFP) is shown in Algorithm ~\ref{Algom 3}.

\begin{algorithm}[htb]
	\caption{CFP} 
	\label{Algom 3}		
	\hspace*{0.02in} \leftline{{\bf Input:} Suffix pattern \textbf{s} of \textbf{o}, occurrences of \textbf{s}: $Occ_\textbf{s}$, keypoint time series \textbf{k}}\\
		\hspace*{0.02in} \leftline{{\bf Output:} Fusion patterns $F_\textbf{o}$, occurrences of $F_\textbf{o}$: $Occ_{F_\textbf{o}}$}\\
		\begin{algorithmic}[1]
		
		\FOR {each \textit{c} in $Occ_\textbf{s}$}
		\STATE $v$←Binary\_search(\textbf{s}, \textit{c}, $k_{c+1}$);
		\STATE $F_{\textbf{o},v}$++;
		\STATE Append $\left\langle \textit{c}+1 \right\rangle$ to $Occ_{F_{\textbf{o},v}}$;
		\ENDFOR
		\RETURN $F_\textbf{o}$ and $Occ_{F_\textbf{o}}$;
	\end{algorithmic}
\end{algorithm}

\subsection{Calculating supports of superpatterns} \label {subsect4.3}

Now, we know the occurrences of prefix OPP and suffix OPP. To calculate the supports of superpatterns, we propose the algorithm CSS based on pattern fusion whose principle is shown as follows.

Suppose $Occ_\textbf{o}$ = $\{go_1, \cdots, go_i, \cdots, go_y\}$ and $Occ_\textbf{f}$ = $\{gf_1, \cdots, gf_j, \cdots, gf_z\}$ are occurrences of OPPs \textbf{o} =  ($o_1, o_2, \cdots, o_{m}$) and \textbf{f} =  ($f_1, f_2, \cdots, f_{m}$), respectively. There are two cases for pattern \textbf{o} fused with pattern \textbf{f}.

\textbf{Case 1}: If $o_1$ = $f_m$, then two superpatterns with length \textit{m}+1, \textbf{u} = ($u_1,\cdots, u_i,\cdots, u_{m+1}$) and \textbf{v} = ($v_1,\cdots, v_i,\cdots, v_{m+1}$) are generated by \textbf{u}, \textbf{v} = \textbf{o} $\oplus$ \textbf{f}, where $u_1=v_{m+1}=o_1$ and $u_{m+1}=v_1=o_1+1$. For $1 \leqslant i \leqslant \textit{m}-1$, if $f_i<o_1$, then $u_{i+1}=v_{i+1}=f_i$; otherwise, if $f_i>o_1$, then $u_{i+1}=v_{i+1}=f_i+1$. Furthermore, the occurrences of \textbf{u} and \textbf{v} are calculated as follows: $\left\langle gf_j \right\rangle$ may be an occurrence of \textbf{u} or \textbf{v}, iff $gf_j$=$go_i$+1. Then, we need to determine the values of $k_{begin}$ and $k_{end}$, where \textit{begin}=$gf_j-m$ and \textit{end}=$gf_j$.

(1) If $k_{begin}=k_{end}$, then $\left\langle gf_j \right\rangle$ is not an occurrence of any superpattern.

(2) If $k_{begin}<k_{end}$, then $\left\langle gf_j \right\rangle$  is an occurrence of \textbf{u}, i.e., $\left\langle gf_j \right\rangle$ is added into $Occ_\textbf{u}$, and $\left\langle go_i \right\rangle$ and $\left\langle gf_j \right\rangle$ are pruned from $Occ_\textbf{o}$ and $Occ_\textbf{f}$, respectively.

(3) If $k_{begin}>k_{end}$, then $\left\langle gf_j \right\rangle$ is an occurrence of \textbf{v}, i.e., $\left\langle gf_j \right\rangle$ is added into $Occ_\textbf{v}$, and $\left\langle go_i \right\rangle$ and $\left\langle of_j \right\rangle$ are pruned from $Occ_\textbf{o}$ and $Occ_\textbf{f}$, respectively.

\textbf{Case 2}: {If $o_1 \neq f_m$, then a superpattern with length \textit{m}+1, \textbf{u} = ($u_1,…, u_i,…, u_{m+1}$) is generated by \textbf{u} = \textbf{o} $\oplus$ \textbf{f}.}

(1) If $o_1<f_m$, then $u_1$=$o_1$. For $1 \leqslant i \leqslant m$, if $f_i<o_1$, then $u_i$+1=$f_i$; otherwise,  $u_i$+1=$f_i$+1.   

(2) If $o_1>f_m$, then $u_1$=$o_1$+1. For $1 \leqslant i \leqslant m$, if $f_i \leqslant o_1$, then $u_i$+1=$f_i$; otherwise, $u_i$+1=$f_i$+1.  

Moreover, the occurrences of pattern \textbf{u} are calculated as follows: $\left\langle gf_j \right\rangle$ is an occurrence of \textbf{u}, iff $gf_j$=$go_i$+1, i.e., $\left\langle gf_j \right\rangle$ is added into $Occ_\textbf{u}$, and $\left\langle go_i \right\rangle$ and $\left\langle gf_j \right\rangle$ are pruned from $Occ_\textbf{o}$ and $Occ_\textbf{f}$, respectively. $sup(\textbf{u}, \textbf{k})$=$|Occ_\textbf{u}|$.

According to Cases 1 and 2, we know that if we find an occurrence of the superpattern, then an element will be pruned from  $Occ_\textbf{o}$. Therefore, the size of $Occ_\textbf{o}$ gradually decreases.  To further improve efficiency, we propose an ending strategy.

\textbf{Ending strategy.} If $|Occ_\textbf{o}| < minsup$, then CSS is terminated.

Example ~\ref{Example 11} illustrates the principle of CSS.

\begin{example}
	In this example, we use the results of Examples \ref{Example 9} and ~\ref{Example 10}. We know that \textbf{o} = (2,1,4,3) and $Occ_\textbf{o}$ = $\{$6,8,12,14$\}$, $F_\textbf{o}$ = $\{$(1,3,2,4)$\}$ and $Occ_{(1,3,2,4)}$ = $\{$5,7,9,13,15$\}$. Pattern \textbf{o} can be fused with each frequent fusion pattern in $F_\textbf{o}$ using CSS. If $minsup$=3, then a pattern that does not meet the threshold can be pruned. Suppose \textbf{f} = $F_{\textbf{o},4}$ = (1,3,2,4). Patterns \textbf{o} and \textbf{f} can be fused into one pattern, i.e., \textbf{o} $\oplus$ \textbf{f} = \textbf{u} = (2,1,4,3,5), since $o_1 \neq f_4$, which belongs to Case 2. We know the first element of $Occ_\textbf{o}$ is 6, and 7 is in $Occ_\textbf{f}$. Thus, $\left\langle\textit{7}\right\rangle$ is an occurrence of \textbf{u}, i.e., $\left\langle\textit{7}\right\rangle$ is added into $Occ_\textbf{u}$, and $\left\langle\textit{6}\right\rangle$ and $\left\langle\textit{7}\right\rangle$ are pruned from $Occ_\textbf{o}$ and $Occ_\textbf{f}$, respectively. Similarly, $\left\langle\textit{9}\right\rangle$, $\left\langle\textit{13}\right\rangle$, and $\left\langle\textit{15}\right\rangle$ are also three occurrences of \textbf{u}; $\left\langle\textit{9}\right\rangle$, $\left\langle\textit{13}\right\rangle$, and $\left\langle\textit{15}\right\rangle$ are added into $Occ_\textbf{u}$, i.e., $Occ_\textbf{u}$ = $\{$7,9,13,15$\}$ and $sup(\textbf{u},\textbf{k})$=4; $\left\langle\textit{8}\right\rangle$, $\left\langle\textit{12}\right\rangle$, and $\left\langle\textit{14}\right\rangle$ are pruned from $Occ_\textbf{o}$; $\left\langle\textit{9}\right\rangle$, $\left\langle\textit{13}\right\rangle$, and $\left\langle\textit{15}\right\rangle$ are pruned from $Occ_\textbf{f}$, i.e., $Occ_\textbf{o}$ = $\{$$\emptyset$$\}$ and $Occ_\textbf{f}$ = $\{$5$\}$. Hence,  \textbf{u} = (2,1,4,3,5) is a frequent COP with length 5, since $sup(\textbf{u},\textbf{k})$=4. Now, CSS is terminated, since the size of $Occ_\textbf{o}$ is zero.
	\label{Example 11}
\end{example}

Now, we prove that the ending strategy is correct. Before this, we first show that COP mining satisfies anti-monotonicity.

\begin{theorem}
COP mining satisfies anti-monotonicity, that is, the support of a superpattern is not greater than that of its subpattern.
	\label{Theorem mono}
\end{theorem}

\begin{proof}
Suppose  $\left\langle gx_j \right\rangle$ is an occurrence of superpattern \textbf{x} = ($x_1, x_2,…, x_m$). Then we can safely say that $\left\langle gx_j -1 \right\rangle$ is an occurrence of subpattern ($x_1, x_2,…, x_{m-1}$) according to Cases 1 and 2. Therefore, the number of occurrences of a superpattern is not greater than that of its subpattern, i.e., the support of a superpattern is not greater than that of its subpattern. Hence, COP mining satisfies anti-monotonicity.

\end{proof}

\begin{theorem}
	The ending strategy is correct.
	\label{Theorem 1}
\end{theorem}

\begin{proof}
	\rm{When pattern \textbf{o} fuses with pattern \textbf{f}, if an occurrence of superpattern \textbf{u} or \textbf{v} is found, then the corresponding occurrence of $Occ_\textbf{o}$ is removed, which means that the size of $Occ_\textbf{o}$ gradually decreases. We know that COP mining satisfies anti-monotonicity. Therefore, when $|Occ_\textbf{o}| < minsup$, if pattern \textbf{o} fuses with a new pattern \textbf{w}, then the support of superpattern is also less than \textit{minsup}. Hence, CSS can be terminated, which means that we do not need to fuse pattern \textbf{o} with the other patterns in $F_\textbf{o}$.}
\end{proof}

The pseudo-code of CSS is shown in Algorithm ~\ref{Algom 4}.

\begin{algorithm}[h]
	\caption{CSS} 
	\label{Algom 4}		
	\hspace*{0.02in} \leftline{{\bf Input:} Keypoint time series \textbf{k}, OPP \textbf{o} of keypoint sub-time series \textbf{q}, occurrences of \textbf{o}: $Occ_\textbf{o}$,}\\ 
 \hspace*{0.43in} \leftline{fusion patterns $F_\textbf{o}$, occurrences of $F_\textbf{o}$: $Occ_{F_\textbf{o}}$, and \textit{minsup}}\\
		\hspace*{0.02in} \leftline{{\bf Output:} Frequent COP set $C_{m+1}$ and their occurrences $Occ_{C_{m+1}}$}\\
		\begin{algorithmic}[1]
		\FOR {each pattern \textbf{f} in $F_\textbf{o}$}
		\IF {$|Occ_\textbf{f}| \geqslant minsup$}
		\IF {\textbf{o}[1] == \textbf{f}[\textit{m}]}
		\STATE \textbf{u}, \textbf{v} = \textbf{o} $\oplus$ \textbf{f};
		\STATE Obtain occurrences $Occ_\textbf{u}$ and $Occ_\textbf{v}$ according to Case 1, and update $Occ_\textbf{o}$ and $Occ_\textbf{f}$;
		\STATE If \textbf{u} and \textbf{v} are frequent, then store them into $C_{m+1}$;
		\ELSE {}
		\STATE \textbf{u} = \textbf{o} $\oplus$ \textbf{f};
		\STATE Obtain occurrences $Occ_\textbf{u}$ according to Case 2, and update $Occ_\textbf{o}$ and $Occ_\textbf{f}$;
		\STATE If \textbf{u} is frequent, then store it into $C_{m+1}$ ;
		\ENDIF
		\ENDIF
		\IF {$|Occ_\textbf{o}|<minsup$}
		\STATE break; // Ending strategy
		\ENDIF
		\ENDFOR
		\RETURN $C_{m+1}$ and $Occ_{C_{m+1}}$;
	\end{algorithmic}
\end{algorithm}

\subsection{COP-Miner} \label {subsect4.4}

To mine all COPs of the given pattern \textbf{o} = \textit{R}(\textbf{q}), we propose COP-Miner, whose main steps are as follows.

Step 1: Extract keypoints of time series \textbf{t} and sub-time series \textbf{p} using the LEP algorithm, and then obtain keypoint time series \textbf{k} and keypoint sub-time series \textbf{q}.

Step 2: Adopt the method of preparation stage to calculate the occurrences of pattern \textbf{o} and further generate all fusion patterns $F_\textbf{o}$, thereby calculating the occurrences of all fusion patterns of pattern \textbf{o}.

Step 3: Use the CSS algorithm to calculate the supports of superpatterns, and then mine frequent COPs $C_{m+1}$ with length \textit{m}+1.

Step 4: For each frequent pattern \textbf{h} in $C_k$, adopt the CFP algorithm to calculate the occurrences of all fusion patterns of pattern \textbf{h}.

Step 5: Iterate Steps 3 and 4, until $C_{k+1}$ is empty.

Example \ref {example12} is used to illustrate the principle of COP-Miner.

\begin{example} \label {example12}
	Suppose we have the same time series \textbf{t} as in Example ~\ref{Example 2} and pattern \textbf{p} = (5,3,7,13,8), and we set \textit{minsup}=3. 
		
		According to Step 1, we obtain keypoint time series \textbf{k} shown in Table ~\ref{tab2} and keypoint sub-time series \textbf{q} = (5,3,13,8) using the LEP algorithm. 
		
		According to Step 2, the OPP of \textbf{q} is \textbf{o} = \textit{R}(\textbf{q}) = (2,1,4,3), and we adopt the method of preparation stage to calculate the occurrences $Occ_\textbf{o}$ = $\{$6,8,12,14$\}$. There is one fusion pattern $F_{\textbf{o},4}$ = (1,3,2,4), whose occurrences is $Occ_{F_{\textbf{o},4}}$ = $\{$5,7,9,13,15$\}$.
		
		According to Step 3, pattern \textbf{o} and fusion pattern $F_{\textbf{o},4}$ can be fused into a superpattern \textbf{u} = (2,1,4,3,5) and $Occ_\textbf{u}$ = $\{$7,9,13,15$\}$ using the CSS algorithm. 
		
		According to Step 4, \textbf{u} is the only frequent COP, and we store \textbf{u} into $C_5$. Using the CFP algorithm, we discover two fusion patterns $F_{\textbf{u},2}$ = (1,4,3,5,2) and $F_{\textbf{u},4}$ = (1,3,2,5,4), whose occurrences are $Occ_{F_{\textbf{u},2}}$ = $\{$10$\}$ and $Occ_{F_{\textbf{u},4}}$ = $\{$8,14$\}$. $C_6$ is empty, since these two fusion patterns of pattern \textbf{u} are infrequent. 
		
		Therefore, all frequent COPs of pattern \textbf{q} are in \textit{C} = $\{$(2,1,4,3,5)$\}$.
	\label{Example 12}
\end{example}

The pseudo-code of COP-Miner is shown in Algorithm ~\ref{Algom 6}.

\begin{algorithm}[htbp]
	\caption{COP-Miner} 
	\label{Algom 6}		  	
	\hspace*{0.02in} \leftline{{\bf Input:} Time series \textbf{t}, pattern \textbf{p}, \textit{minsup}}\\
		\hspace*{0.02in} \leftline{{\bf Output:} Frequent COPs \textit{C}}\\
		\begin{algorithmic}[1]
		
		\STATE \textbf{k}←LEP(\textbf{t}), \textbf{q}←LEP(\textbf{p});
		//Preparation stage
		\STATE \textbf{o} = \textit{R}(\textbf{q}), calculate \textbf{s} = \textit{suffixorder}(\textbf{o}) according to Definition ~\ref{Definition 7};
		\STATE Use FAV to calculate the occurrences $Occ_\textbf{s}$;
		\STATE $Occ_\textbf{o}$←VOP(\textbf{o}, \textbf{s}, $Occ_\textbf{s}$, \textbf{k});
		\STATE $F_\textbf{o}$, $Occ_{F_\textbf{o}}$←CFP(\textbf{o}, $Occ_\textbf{s}$, \textbf{k});
		
		
		//Iteratively calculating supports and mining frequent COPs
		\STATE $C_{m+1}, Occ_{C_{m+1}}$←CSS($\textbf{k}, \textbf{o}, Occ_\textbf{o}, F_\textbf{o}, Occ_{F_\textbf{o}}, minsup$);
		\STATE \textit{g}←\textit{m}+1, \textit{C}←$C_{m+1}$;
		\WHILE {$C_\textit{g} \neq $ null}
		\FOR {each \textbf{h} in $C_\textit{g}$}
		\STATE $F_\textbf{h}, Occ_{F_\textbf{h}}$←CFP(\textbf{h}, $Occ_{{suffixorder}(\textbf{h})}$, \textbf{k});
		\STATE $C_{\textit{g}+1}, Occ_{C_{\textit{g}+1}}$←CSS($\textbf{k}, \textbf{h}, Occ_\textbf{h}, F_\textbf{h}, Occ_{F_\textbf{h}}, minsup$);\STATE \textit{C}←$C \cup\ C_{\textit{g}+1}$;
		\ENDFOR
		\STATE \textit{g}←\textit{g}+1;
		\ENDWHILE
		\RETURN \textit{C};
	\end{algorithmic}
\end{algorithm}

\begin{theorem}
\textcolor{black}{The time complexity of COP-Miner is $O(N+M+n \times log_2 ^m+d\times (n+m))$, where $N$, $M$, $n$, $m$, and $d$ are the lengths of the original time series \textit{t}, the given prefix pattern \textit{p}, the keypoint time series \textit{k}, the keypoint pattern \textit{q}, and the number of iterations of Line 8 in COP-Miner which means that the maximum length of the mined pattern is d+m, respectively.}
\label{Theorem 2}
\end{theorem}

\begin{proof}
\textcolor{black}{ Line 1 of COP-Miner extracts the keypoints of the original time series and the given prefix pattern. The time complexity of Line 1 is $O(N+M)$. Line 2 calculates \textit{suffixorder}(\textbf{q}) whose time complexity is $O(m)$. The time complexity of generating binary sequences in FAV is $O(n+m)$. The time complexity of SBNDM$_2$ matching in FAV is $O(n \times log_2^m)$. The time complexity of verification in FAV is $O(n)$. Thus, the time complexity of Line 3 is $O(n\times log_2^m+m)$. The maximal length of Occs is less than $n$. Thus, the time complexity of Line 4 (the VOP algorithm) is $O(n)$. Moreover, the time complexity of Line 2 in the CFP algorithm is $O(log_2^m)$. Therefore, the time complexity of Line 5 (the CFP algorithm) is $O(n \times log_2^m)$. The CSS algorithm has two parts: generating superpatterns and calculating their occurrences. Thus, the time complexity of Line 6 (the CSS algorithm) is $O(n+m)$. The time complexity of Lines 8 to 15 is $O(d\times (n+m))$, since the number of iterations of Line 8 is $d$. Hence, the time complexity of COP-Miner is $O(N+M+n \times log_2 ^m+d\times (n+m))$.}

\end{proof}

\begin{theorem}
\textcolor{black}{The space complexity of COP-Miner is  $O(N+M+ d \times (n+m))$.}
\label{Theorem 3}
\end{theorem}

\begin{proof}
\textcolor{black}{The space complexity of Line 1 of COP-Miner is $O(N+M+n+m)$, since LEP converts \textbf {t} and \textbf {p} into \textbf {k} and \textbf {q}, respectively. The space complexity of Line 2 is $O(m)$. The space complexity of Line 3 is $O(n)$. The space complexity of Line 4 (the VOP algorithm) is $O(n)$. The space complexity of Line 5 (the CFP algorithm) is $O(n+m)$. Moreover, the space complexity of Lines 8 to 15 is $O(d \times (n+m))$. Hence, the space complexity of COP-Miner is $O(N+M+ d \times (n+m))$.}

\end{proof}

\section{Experimental results and analysis} \label {section5}

In this section, we consider the performance of COP-Miner from multiple perspectives. Moreover, we answer the following five research questions (RQs) through experiments:

RQ1: Do we need to specially design an algorithm?

RQ2: What is the effect of each strategy in COP-Miner?

RQ3: Does COP-Miner have better scalability than other competing algorithms?

RQ4: How do different parameters affect the performance of COP-Miner?

RQ5: What is the impact of keypoint extraction on time series prediction?

To answer RQ1, we compared the efficiency of COP-Miner with OPP-Miner and EFO-Miner, which mine all OPPs and then filter COPs. To address RQ2, we proposed COP-kom, COP-noVOP, COP-noCFP, COP-efo, COP-enum, COP-noending, and FAV-sliding to verify the efficiency of COP-Miner (see Section \ref {subsect5.3}). To answer RQ3, we conducted experiments on datasets of different lengths to report the performance of scalability (see Section \ref {subsect5.4}). To address RQ4, we explored the effects of prefix patterns with different lengths and minimum support thresholds on the running time of the algorithm (see Section \ref {subsect5.5}). To solve RQ5, we selected two models: autoregressive integrated moving average (ARIMA)  \cite{ARIMA, LiARIMA2022}, dynamic time warping (DTW) \cite{KeoghDTW2005}, and proposed COP-noLEP to verify the prediction performance(see Section \ref {subsect5.6}).

\subsection{Datasets} \label {subsect5.1}

We used nine real datasets for experimental evaluation. The electrocardiogram (ECG) datasets include KURIAS\_HeartRate, KURIAS\_PAxis, and KURIAS\_PRInterval. They are part of a 12-lead ECG database with standardized diagnostic ontology, which can be downloaded at https://physionet.org /content/kurias-ecg/1.0/KURIAS-ECG. The stock datasets are the opening price of the S\&P 500 and listed stocks on the New York Stock Exchange, which can be downloaded at https://www.yahoo.com/. 
To validate the performance of COP-Miner on large-scale dataset, we generated a synthetic real dataset named Big\_NYSE by enlarging SDB8 440 times. The PM2.5 and temperature datasets are the records of Chengdu and Guangzhou, respectively, which can be downloaded at https://archive.ics.uci. edu/ml/datasets.php. The CSSE COVID-19 dataset is the information on the number of confirmed cases in Alabama, USA, which can be downloaded at https://datahub.io/core/covid-19. The specific description of each dataset is shown in Table ~\ref{tab3}.

\begin{table}[htbp]
	\renewcommand{\arraystretch}{1.1}    
		\footnotesize %
		\captionsetup{font=footnotesize} %
		\caption{Description of datasets} %
		\centering %
		\label{tab3} %
		\tabcolsep 12pt 
			\begin{tabular}{lccc} %
				\toprule %
				Dataset & Name & Type & Total length \\
				\midrule
				SDB1 & KURIAS\_HeartRate & ECG & 13,863 \\
				SDB2 & KURIAS\_PAxis & ECG & 13,863 \\
				SDB3 & KURIAS\_PRInterval & ECG & 13,863 \\
				SDB4 & S\&P 500 & Stock & 23,046 \\
				SDB5 & ChengduPM2.5 & PM2.5 & 28,900 \\
				SDB6 & GuangZhouTemp & Temperature & 52,582 \\
				SDB7 & US\_Alabama\_Confirmed & CSSE COVID-19 & 56,304 \\
				SDB8 & NYSE & Stock & 60,000 \\
				{SDB9} & {Big\_NYSE} & {Stock} & {26,400,000} \\
				\bottomrule
		\end{tabular}
		\label{tab3} %
	\end{table}
	
	We ran the experiments presented in this paper on a computer with Intel(R) Core(TM) i5-5200U, 2.20GHZ CPU, 4.0GB RAM, and the Windows 10 64-bit operating system. The programming environment was IntelliJ IDEA 2019.2. \textcolor {black}{All the algorithms and datasets can be downloaded from     https://github.com/wuc567/Pattern-Mining/tree/master/COP-Miner.}
	
	\subsection{Baseline methods}
	To evaluate the efficiency and prediction performance of the COP-Miner algorithm, we designed four experiments and selected ten competing algorithms for comparison.
	
	1. OPP-Miner \cite{WuOPP-Miner2022} and EFO-Miner \cite{WuOPR-Miner2022}: To validate that it is necessary to design a special algorithm, we applied OPP-Miner and EFO-Miner, which adopt a pattern matching strategy and subpattern’s matching results to calculate the support of superpattern, respectively. They both generate COPs based on find-all-frequent OPPs.
	
	2. COP-kom: To validate the efficiency of the FAV \cite{WuOPP-Miner2022} algorithm, we proposed COP-kom, which applies the KMP-Order-Matcher \cite{Kimmatching2014} algorithm to obtain the occurrences of \textit{{suffixorder}}(\textbf{q}).
	
	3. COP-noVOP: To validate the efficiency of the VOP algorithm, we developed COP-noVOP, which does not use the VOP algorithm, but rather applies a linear search to obtain the occurrences of pattern \textbf{q}.
	
	4. COP-noCFP: To validate the efficiency of CFP algorithm, we developed COP-noCFP, which does not use a binary search, but rather applies a linear search to calculate the occurrences of all fusion patterns.
	
	5. COP-efo: To evaluate the superiority of the method of the preparation stage, we proposed COP-efo, which employs EFO-Miner \cite{WuOPR-Miner2022} to calculate the occurrences of prefix pattern \textbf{q} and its fusion patterns.
	
	6. COP-enum: To analyze the performance of pattern fusion strategy in CSS, we proposed COP-enum algorithm, which generates candidate patterns using an enumeration method.
	
	7. COP-noending: To verify the effectiveness of the ending strategy, we developed COP-noending, which adopts the CSS algorithm to calculate the support, but does not use the ending strategy.
	
	8. FAV-sliding: To verify the efficiency of the CSS algorithm, we designed FAV-sliding, which employs FAV \cite{WuOPP-Miner2022} to calculate the occurrences $Occ_\textbf{s}$ of pattern \textbf{s}, and then uses the sliding windows to calculate the occurrences of superpatterns.
	
	9. COP-noLEP: To verify the influence of extracting keypoints on the prediction performance, we developed COP-noLEP to generate COPs, which does not use the LEP algorithm, but rather mines COPs directly from the original time series.

	\subsection{Performance of COP-Miner} \label {subsect5.3}
	
	To determine the necessity of specially designing COP-Miner, we selected OPP-Miner and EFO-Miner as the competing algorithms. To evaluate the performance of the FAV, VOP, and CFP algorithms, we employed COP-kom, COP-noVOP, and COP-noCFP as competing algorithms. To evaluate the performance of the method of the preparation stage, we designed COP-efo as a competing algorithm. To evaluate the performance of the pattern fusion strategy and ending strategy, we employed  COP-enum and COP-noending as competing algorithms. Moreover, to examine the performance of the CSS algorithm, we designed FAV-sliding as a competing algorithm. We set \textit{minsup}=10 on SDB1-SDB8 and set \textit{minsup}=6000 on SDB9. Since all ten algorithms are complete, they mine the same number of frequent patterns. When the length of a given pattern is six, i.e., $\textbf{p}_{\rm{SDB1}}$ = $\textbf{p}_{\rm{\rm{SDB3}}}$ = (1,5,2,6,3,4), $\textbf{p}_{\rm{SDB2}}$ = (1,5,3,4,2,6), $\textbf{p}_{\rm{SDB4}}$ = (1,4,3,6,5,2), $\textbf{p}_{\rm{SDB5}}$ = $\textbf{p}_{\rm{SDB7}}$ = $\textbf{p}_{\rm{SDB8}}$ = {$\textbf{p}_{\rm{SDB9}}$ =} (1,3,2,5,4,6), and $\textbf{p}_{\rm{SDB6}}$ = (3,6,4,5,1,2), there are 5, 5, 6, 24, 21, 2, 2, {82, and 52 frequent COPs mined from SDB1-SDB9,} respectively. Figs. ~\ref{RunningTime} and \ref{MemoryUsage} show the comparisons of running time and memory usage, respectively. We also show the comparisons of the numbers of candidate patterns and the number of occurrences of superpatterns in Figs. ~\ref{CandidatesNumber} and \ref{ElementsComparison}, respectively. It is worth noting that the information regarding the OPP-Miner and FAV-sliding algorithms are not presented in Fig. \ref{ElementsComparison}, since they do not use the occurrences of subpatterns to calculate the support of superpattern.
 
\begin{figure}[htbp]
	
	\centering
	\includegraphics[width=0.65\linewidth]{"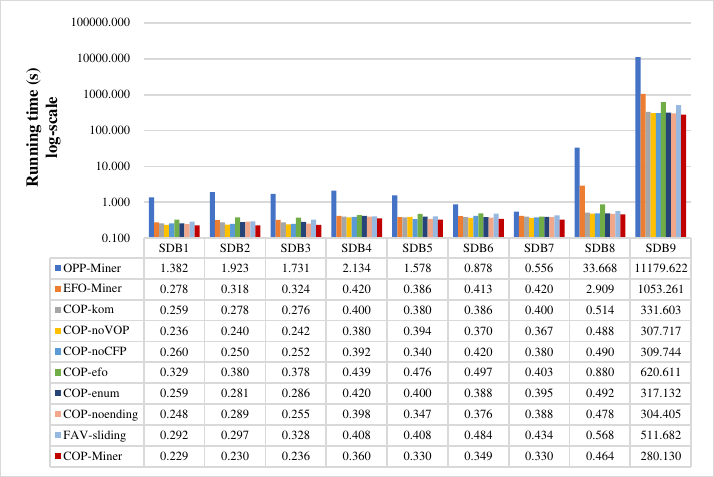"}
	\caption{\textcolor {black}{Comparison of running time on SDB1-SDB9}}
	\label{RunningTime}
\end{figure}

\begin{figure}[htbp]
	\centering
	\includegraphics[width=0.65\linewidth]{"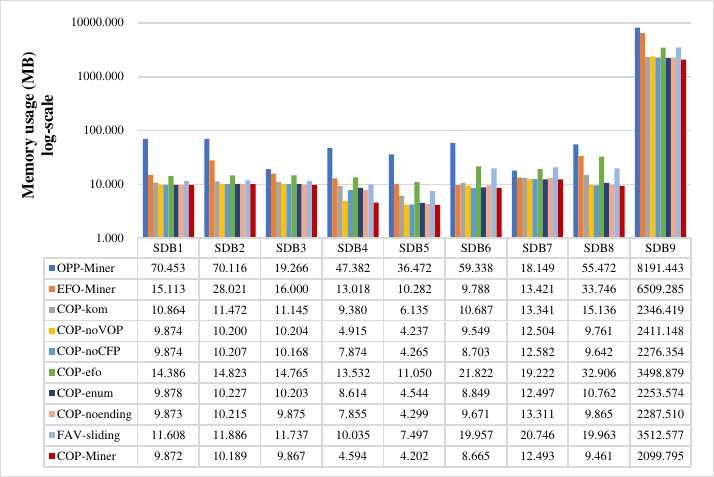"}
	\caption{\textcolor {black}{Comparison of memory usage on SDB1-SDB9}}
	\label{MemoryUsage}	
\end{figure}

\begin{figure}[htbp]
	
	\centering
	\includegraphics[width=0.65\linewidth]{"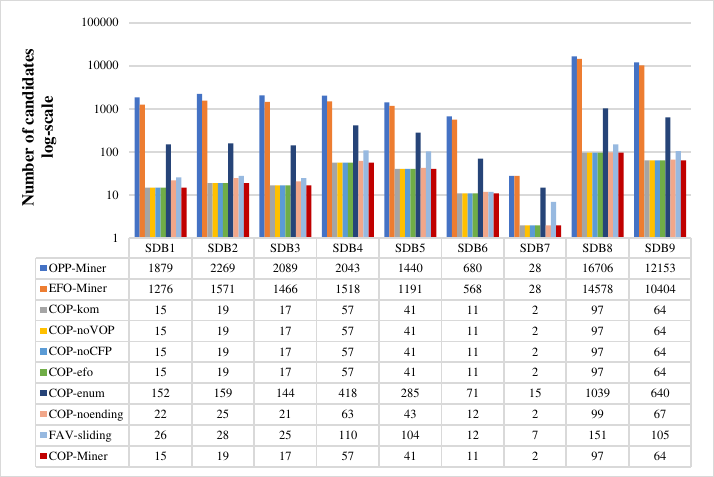"}
	\caption{\textcolor {black}{Comparison of numbers of candidate patterns on SDB1-SDB9}}
	\label{CandidatesNumber}
\end{figure}

\begin{figure}[htbp]
	
	\centering
	\includegraphics[width=0.65\linewidth]{"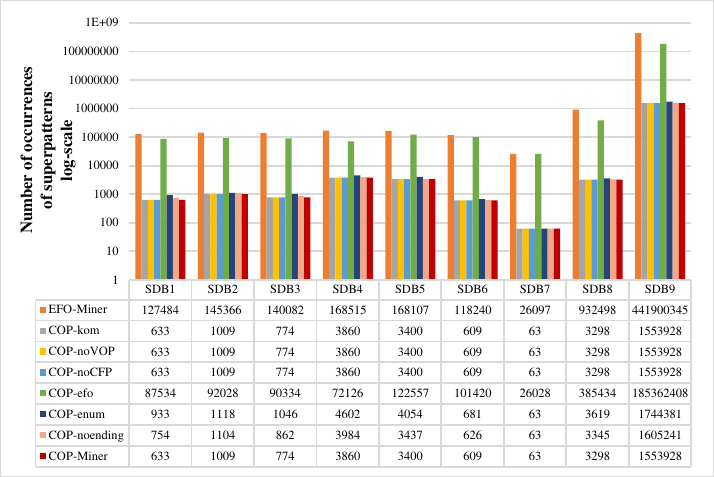"}
	\caption{\textcolor {black}{Comparison of numbers of occurrences of superpatterns on SDB1-SDB9}}
	\label{ElementsComparison}
\end{figure}

The results give rise to the following observations.

1. COP-Miner runs faster and costs less memory than both OPP-Miner and EFO-Miner. For example, on {SDB9}, it can be seen from Fig. ~\ref{RunningTime} that OPP-Miner and EFO-Miner require {11,179.622s} and {1,053.261s} of running time, respectively, while COP-Miner needs {280.130s}. Thus, COP-Miner is about 40 times and 4 times faster than OPP-Miner and EFO-Miner, respectively. In Fig. ~\ref{MemoryUsage}, OPP-Miner and EFO-Miner occupy {8,191.4Mb} and {6,509.3Mb} of running memory, respectively, while COP-Miner occupies only {2,099.8Mb}. The reason is as follows. Both OPP-Miner and EFO-Miner discover COPs based on finding all frequent OPPs, which is a brute-force method to discover COPs. Thus, the two algorithms have to check more candidate patterns and generate more occurrences of superpatterns. Figs. ~\ref{CandidatesNumber} and ~\ref{ElementsComparison} also verify these claims. From Fig. ~\ref{CandidatesNumber}, OPP-Miner and EFO-Miner generate {12,153} and {10,404} candidate patterns, respectively, while COP-Miner generates only {64}. From Fig. ~\ref{ElementsComparison}, EFO-Miner makes {441,900,345} comparisons while COP-Miner makes only {1,553,928}. We can find the same effect on other datasets. Therefore, COP-Miner outperforms OPP-Miner and EFO-Miner, and it is necessary to specially design an algorithm for COP mining.


2. COP-Miner runs faster and requires less memory than COP-kom, COP-noVOP, and COP-noCFP. For example, on SDB4, it can be seen from Fig. ~\ref{RunningTime} that COP-kom, COP-noVOP, and COP-noCFP require 0.400s, 0.380s, and 0.392s, respectively, while COP-Miner only needs 0.360s. From Fig. ~\ref{MemoryUsage}, COP-kom, COP-noVOP, and COP-noCFP occupy 9.380Mb, 4.915Mb, and 7.874Mb, respectively, while COP-Miner occupies only 4.594Mb. The reasons are as follows. The four algorithms employ the same CSS algorithm to calculate the supports of superpatterns. Thus, from Figs. ~\ref{CandidatesNumber} and ~\ref{ElementsComparison}, we know that the four algorithms generate the same number of candidate patterns and occurrences of superpatterns. For example, on SDB4, the four algorithms generate 57 candidate patterns and make 3,860 comparisons. The difference between COP-kom and COP-Miner is that COP-kom uses the KMP-Order-Matcher \cite{Kimmatching2014} algorithm to find the occurrences of \textit{{suffixorder}}(\textbf{q}), while COP-Miner adopts the FAV algorithm. The results indicate that FAV is more effective than KMP-Order-Matcher. Similarly, the results validate that the designed VOP algorithm is more efficient than the linear search to obtain the occurrences of pattern \textbf{q}, and the binary search in CFP is more effective than the linear search to calculate the occurrences of fusion patterns. Therefore, COP-Miner outperforms COP-kom, COP-noVOP, and COP-noCFP.	

3. COP-Miner runs faster and consumes less memory than COP-efo. For example, on SDB6, from Fig. ~\ref{RunningTime}, COP-efo requires 0.497s, while COP-Miner only needs 0.349s. In Fig. ~\ref{MemoryUsage}, COP-efo occupies 21.822Mb, while COP-Miner occupies only 8.665Mb. The reason is as follows. COP-efo employs EFO-Miner to calculate the occurrences of prefix pattern \textbf{q} and its fusion patterns, which requires pattern fusion and multiple support calculations. Thus, COP-efo has to generate more occurrences of superpatterns. Fig. ~\ref{ElementsComparison} also validates this claim. On SDB6, COP-efo makes 101,420 comparisons, while COP-Miner makes only 609. Moreover, from Fig. ~\ref{CandidatesNumber}, COP-Miner and COP-efo each generate 11 candidate patterns, since they all apply the CSS algorithm to calculate the supports of superpatterns. The same effect can be found on all other datasets. Hence, the method of preparation stage is more efficient than EFO-Miner in calculating the occurrences of pattern \textbf{q} and its fusion patterns, and COP-Miner outperforms COP-efo.

4. COP-Miner runs faster and requires less memory than both COP-enum and COP-noending. For example, on SDB8, from Fig. ~\ref{RunningTime}, COP-enum and COP-noending require 0.492s and 0.478s of running time, respectively, while COP-Miner needs 0.464s. In Fig. ~\ref{MemoryUsage}, COP-enum and COP-noending occupy 10.762MB and 9.865MB of running memory, respectively, while COP-Miner occupies only 9.461Mb. The reasons are as follows. The enumeration method generates more candidate patterns than the pattern fusion strategy. Fig. ~\ref{CandidatesNumber} also validates this claim. On SDB8, COP-enum checks 1,039 candidate patterns, while COP-Miner only checks 97. The ending strategy can further reduce the number of occurrences of superpatterns. Fig. ~\ref{ElementsComparison} also validates this claim. On SDB8, COP-noending makes 3,345 comparisons, while COP-Miner makes 3,298. The same effect can be found on all other datasets. Therefore, COP-Miner outperforms COP-enum and COP-noending.

5. COP-Miner runs faster and consumes less memory than FAV-sliding. For example, on SDB2, it can be seen from Fig. ~\ref{RunningTime} that FAV-sliding requires 0.297s, while COP-Miner needs 0.230s. In Fig. ~\ref{MemoryUsage}, FAV-sliding occupies 11.886Mb, while COP-Miner occupies 10.189Mb. The reason is as follows. We know that FAV-sliding adopts sliding windows to calculate the occurrences of all possible superpatterns. This method can be seen as a linear search to find the co-occurrence patterns. Thus, FAV-sliding has to check more candidate patterns. Fig. ~\ref{CandidatesNumber} also validates this claim. On SDB2, FAV-sliding checks 28 candidate patterns, while COP-Miner only checks 19. The same effect can be found on all other datasets. Hence, COP-Miner outperforms FAV-Sliding.

In summary, the experimental results verify that COP-Miner yields better performance than the other nine competing algorithms.

\subsection{Scalability} \label {subsect5.4}

To evaluate the scalability of COP-Miner, we selected SDB8 with length 60K as the experimental dataset. Moreover, we created 120K, 180K, 240K, 300K, 360K, 420K, and 480K, which were two, three, four, five, six, seven, and eight times the size of SDB8, respectively. We set \textit{minsup}=30 and the prefix pattern was (1,3,2,5,4,6). Comparisons of running time and memory usage are shown in Figs. ~\ref{ScalabilityRunning} and ~\ref{ScalabilityMemory}, respectively.

\begin{figure}[htbp]
	\centering
	\includegraphics[width=0.65\linewidth]{"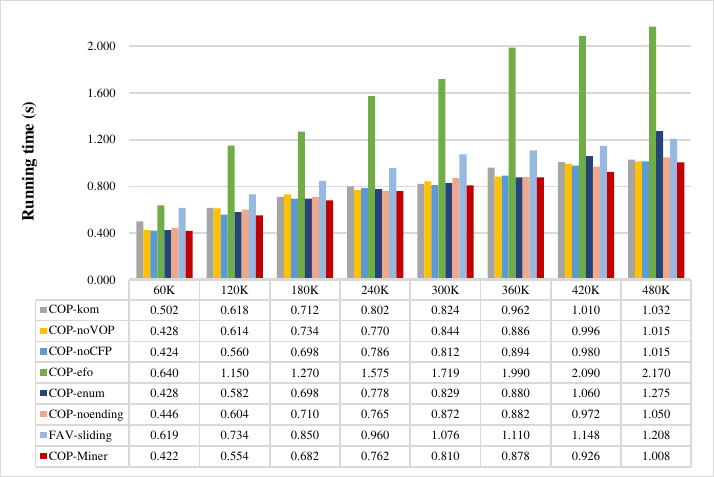"}
	\caption{Comparison of running time with different dataset sizes}
	\label{ScalabilityRunning}
\end{figure}

\begin{figure}[htbp]
	\centering
	\includegraphics[width=0.65\linewidth]{"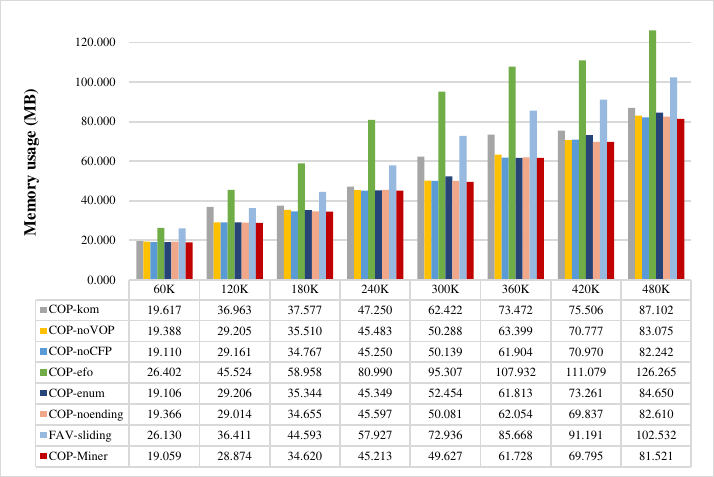"}
	\caption{Comparison of memory usage with different dataset sizes}
	\label{ScalabilityMemory}
\end{figure}

The results give rise to the following observations.

From Figs. ~\ref{ScalabilityRunning} and ~\ref{ScalabilityMemory}, we know that both the running time and memory usage of COP-Miner show linear growth with the increase of the size of the dataset. For example, the size of 480K is eight times of SDB8. COP-Miner takes 0.422s on 60K, and 1.008s on 480K, giving 1.008/0.422=2.389. The memory usage of COP-Miner is 19.059MB on 60K, and 81.521MB on 480K, giving 81.521/19.059=4.277. Thus, the running time and memory usage grow significantly slower than the increase in the dataset size. These results indicate that the time and space complexities of COP-Miner are positively related to the dataset size. This phenomenon can be found on all other datasets. More importantly, we can see that COP-Miner is about two times faster than COP-efo, and the memory usage of COP-efo also exceeds COP-Miner. Therefore, we conclude that COP-Miner has strong scalability, since the mining performance does not degrade as the dataset size increases.

\subsection{Influence of parameters} \label {subsect5.5}

We evaluated the effects of prefix patterns with different lengths and \textit{minsup} on running time. To examine the influence, we selected SDB8 as the experimental dataset, which is the largest dataset in Table ~\ref{tab3}, and selected COP-kom, COP-noVOP, COP-noCFP, COP-efo, COP-enum, and COP-noending as the competing algorithms.

1) Influence of prefix patterns with different lengths:

\textcolor {black}{ To report the influence of prefix patterns with different lengths on running time, we set \textit{minsup}=12. The prefix patterns on SDB8 are (1,3,2), (1,3,2,4), (1,3,2,5,4), (1,3,2,5,4,6), and (1,4,2,6,5,7,3) with lengths 3, 4, 5, 6, and 7, respectively, and all the seven algorithms discover 2121, 1059, 94, 64, and 36 COPs for the five prefix patterns with different lengths, respectively. The comparisons of running time and numbers of candidate patterns on SDB8 are shown in Figs. ~\ref{lenruntime} and \ref {lencandidate}, respectively.}

\begin{figure}[htbp]
	\centering
	\includegraphics[width=0.65\linewidth]{"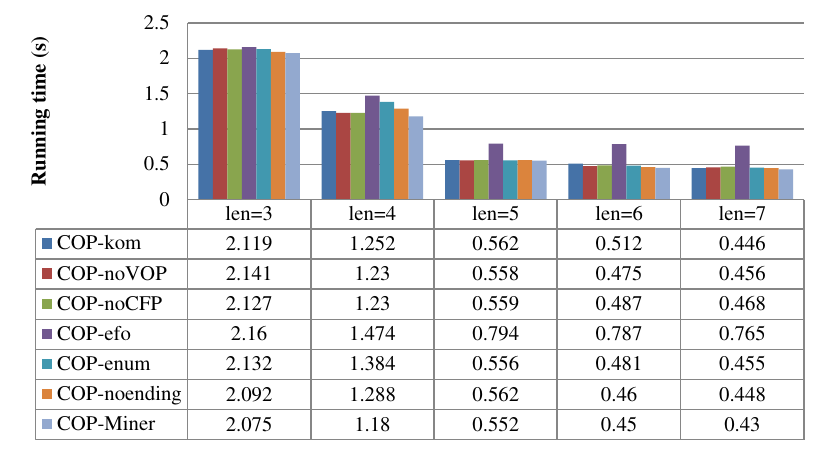"}
	\caption{Comparison of running time of prefix patterns with different lengths on SDB8}
	\label{lenruntime}
\end{figure}

\begin{figure}[htbp]
	\centering
	\includegraphics[width=0.65\linewidth]{"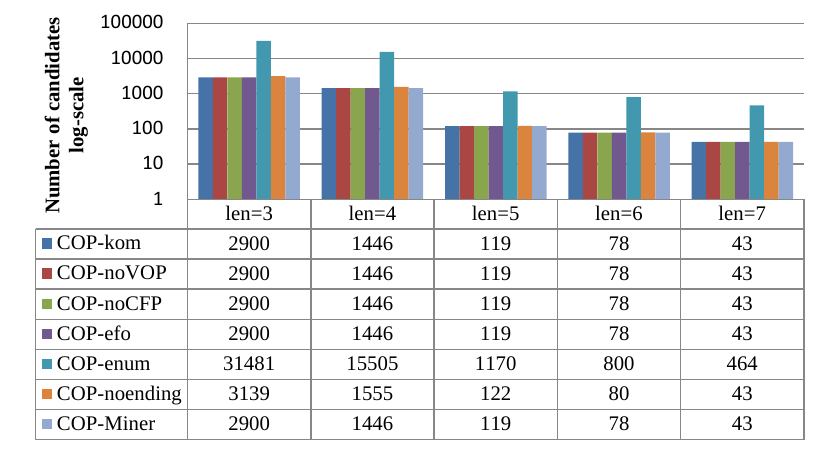"}
	\caption{\textcolor {black}{Comparison of numbers of candidate patterns of prefix patterns with different lengths on SDB8}}
	\label{lencandidate}
\end{figure}

The results give rise to the following observations.

\textcolor {black}{As the length of the prefix pattern increases, the running time decreases, since the numbers of COPs and candidate patterns decrease. For example, from Figs. ~\ref{lenruntime} and \ref {lencandidate}, when len=3, COP-Miner takes 2.075s, checks 2900 candidate patterns, and discovers 2121 COPs, whereas when len=7, COP-Miner takes 0.430s, checks 43 candidate patterns, and discovers 36 COPs. This phenomenon can be found on all other algorithms. The reason is as follows. As the length of the prefix pattern increases, fewer patterns can be found, thereby reducing the number of candidate patterns generated. We know that the fewer the candidate patterns, the shorter the running time. Thus, the running time also decreases. More importantly, COP-Miner outperforms other competing algorithms for any length prefix pattern, since it is the fastest algorithm, as explained in the analysis in Section \ref {subsect5.3}.}

2) Influence of different \textit{minsup}:

\textcolor {black}{To report the influence of different \textit{minsup} on running time, we set \textit{minsup}=8, 10, 12, 14, 16. The prefix pattern on SDB8 is (1,3,2,5,4,6). The seven algorithms discover 97, 82, 64, 48, and 39 COPs for five different \textit{minsup}, respectively. The comparisons of running time and numbers of candidate patterns on SDB8 are shown in Figs. ~\ref{minsupruntime} and \ref{minsupcandidate}, respectively.}

\begin{figure}[htbp]
	\centering
	\includegraphics[width=0.65\linewidth]{"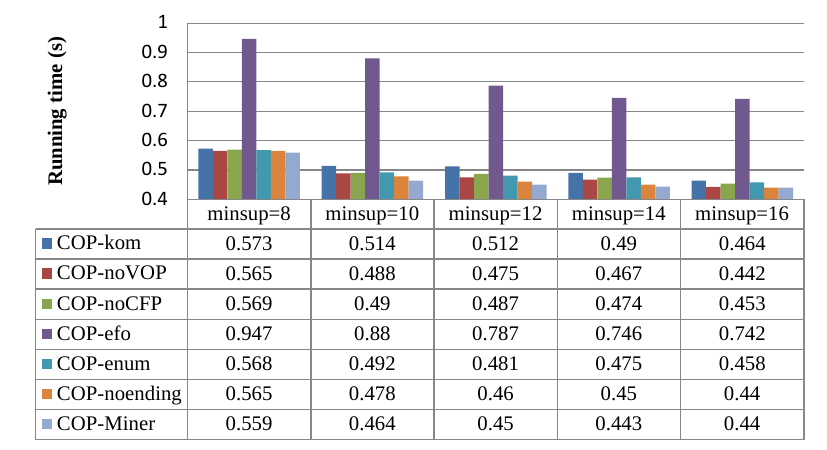"}
	\caption{Comparison of running time with different \textit{minsup} on SDB8}
	\label{minsupruntime}
\end{figure}

\begin{figure}[htbp]
	\centering
	\includegraphics[width=0.65\linewidth]{"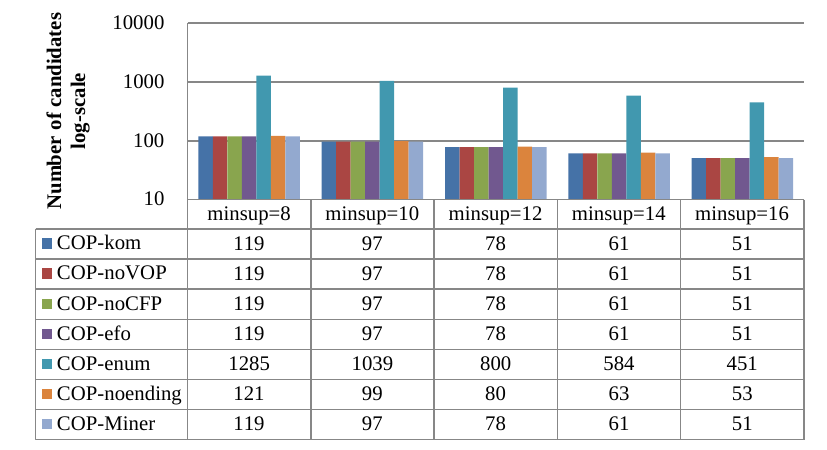"}
	\caption{Comparison of numbers of candidate patterns with different \textit{minsup} on SDB8}
	\label{minsupcandidate}
\end{figure}

The results give rise to the following observations.

\textcolor {black}{As the \textit{minsup} increases, the running time decreases, since the numbers of candidate patterns and COPs decrease. For example, from Figs. ~\ref{minsupruntime} and \ref{minsupcandidate}, when \textit{minsup}=8, COP-Miner takes 0.559s, checks 119 candidate patterns, and discovers 97 COPs, whereas when \textit{minsup}=16, COP-Miner takes 0.440s, checks 51 candidate patterns, and discovers 39 COPs. This phenomenon can be found on all other algorithms. The reason for this is as follows. We know that for a certain prefix pattern, fewer patterns can be frequent patterns and candidate patterns as \textit{minsup} increases. Hence, fewer COPs can be found and the running time also decreases. Moreover, COP-Miner outperforms other competing algorithms for any \textit{minsup}, since it is the fastest algorithm, as explained in the analysis in Section \ref {subsect5.3}.}

\subsection{Case study} \label {subsect5.6}

In this section, we predict intervals of future values based on trend patterns found in historical time series by using COP mining. Fig. ~\ref{TrendPredict} gives the specific prediction by using COPs obtained on the training set of SDB6. The left side of the vertical line is the prefix pattern \textbf{p} = (1,3,2,4) obtained from the historical data. The right side of the vertical line is based on the mined COPs to predict the trend that may occur in the next stage. {We selected the two patterns with the highest frequency and the second highest frequency on the training set to predict the possible intervals at the sixth time point. The two dotted lines represent the possible interval, and the results predicted that the sixth hour should be warmer than 29.3 degrees Celsius. The solid line represents the actual trend and the actual temperature at that time was 29.8.}

\begin{figure}[htbp]
	\centering
	\includegraphics[width=0.65\linewidth]{"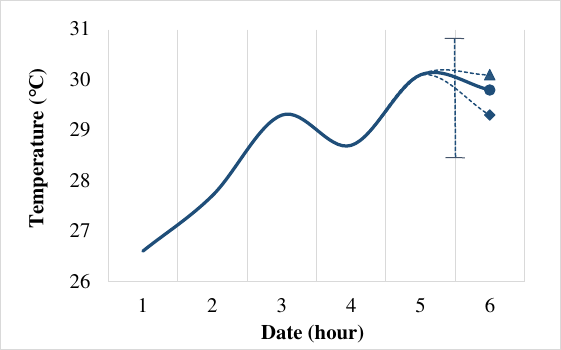"}
	\caption{{COPs of pattern \textbf{p} = (1, 3, 2, 4) on the training set of SDB6}}
	\label{TrendPredict}
\end{figure}

To compare the prediction performance of COP-Miner, we selected two commonly used models: ARIMA  \cite{LiARIMA2022} and DTW \cite{KeoghDTW2005} as comparison models, which predict the OPPs of subsequences similar to the prefix pattern on the original time series. Moreover, we also employed COP-noLEP as a competing algorithm to generate COPs, which does not use the LEP algorithm to eliminate the distortion of time series. We selected SDB1-{SDB9} to evaluate the prediction performance and divided each dataset into a training set and a test set in a ratio of 8:2. The prefix patterns used for training are shown in Table ~\ref{tab4}.

\begin{table}[htbp]
	\centering
		\scriptsize
		\caption{{Prefix patterns}} 
		\label{tab4} 
		\resizebox{9cm}{!}{
			\begin{tabular}{lclclc} 
				\toprule 
				Dataset & Prefix pattern & Dataset & Prefix pattern & Dataset & Prefix pattern \\\hline
				SDB1 & (1,3,2,4) & SDB2 & (1,4,2,3) & SDB3 & (3,2,4,1)\\
				SDB4 & (3,1,5,2,4) & SDB5 & (3,1,2,4,5) & SDB6 & (1,4,2,3)\\
				SDB7 & (1,3,2) & SDB8 & (3,2,5,1,4) & {SDB9} & {(2,4,3,6,5,8,1,7)}\\
				\bottomrule
		\end{tabular}}
	\end{table}
	
	To evaluate the performance of prediction based on the co-occurrence patterns, there are three commonly used criteria: precision \textit{Pr} = $\frac{TP}{TP+FP}$, recall \textit{Re} = $\frac{TP}{TP+FN}$, and F1-score \textit{F}1 = $\frac{{2 \times Pr \times Re}}{Pr+Re}$, where \textit{TP}, \textit{FP}, and \textit{FN} are the numbers of the correctly predicted patterns, wrongly predicted patterns, and unpredicted patterns in the test set, respectively. There is no wrongly predicted pattern, i.e. \textit{FP}=0, which means that the frequent co-occurrence patterns mined in the training set all appear in the test set of these nine datasets. Thus, \textit{Pr}=1. Therefore, we did not choose precision as the criterion, but rather chose recall and F1-score. Recall can analyze the weight of the co-occurrence patterns mined through the training set in the test set. F1-score is a comprehensive consideration of precision and recall, which can more comprehensively evaluate the accuracy of the prediction results. Taking SDB8 as an example, we selected the top two COPs in the training set for COP-Miner: {(4,3,6,1,5,2), (4,3,6,2,5,1)}. {This result means that the value of the sixth point should be smaller than the corresponding true value of rank 3.} In the test set, all length-6 patterns with the prefix pattern (3,2,5,1,4) are (3,2,6,1,5,4), (4,3,6,1,5,2), (4,2,6,1,5,3), and (4,3,6,2,5,1), and their supports are 20, 118, 54, and 132, respectively. The sum of supports of patterns (4,3,6,1,5,2), (4,3,6,2,5,1) and total supports are 118+132=250 and 324, respectively. Thus, recall and F1-score of COP-Miner on the test set are \textit{Re}=$\frac{250}{324}$=0.7716 and \textit{F}1=$\frac{{2 \times 0.7716}}{1+0.7716}$=0.8711. Comparison of recall and F1-score results are shown in Figs. ~\ref{Recall} and \ref{F1score}, respectively.
	
	\begin{figure}[htbp]
		\centering
		\includegraphics[width=0.65\linewidth]{"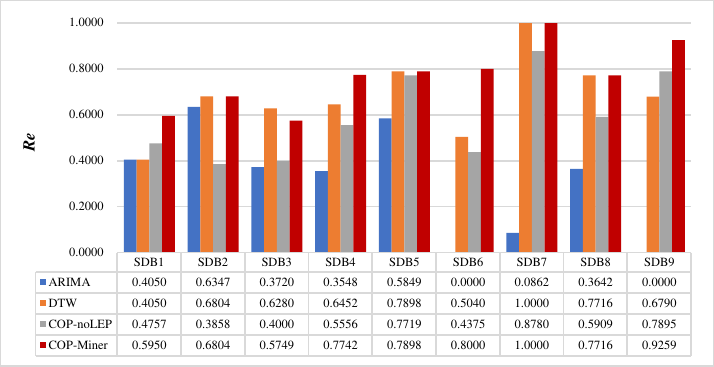"}
		\caption{{Comparison of recall}}
		\label{Recall}
	\end{figure}
	
	\begin{figure}[htbp]
		\centering
		\includegraphics[width=0.65\linewidth]{"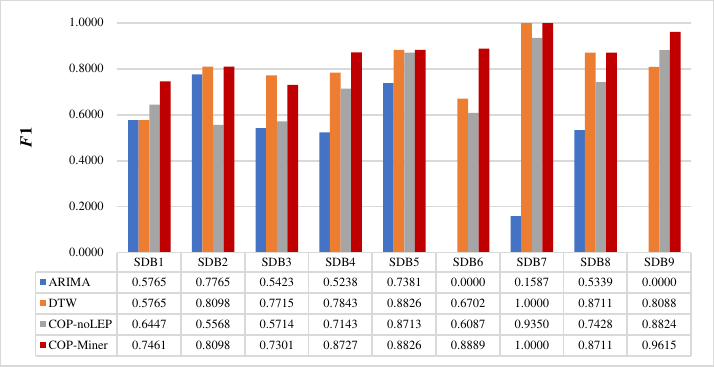"}
		\caption{{Comparison of F1-score}}
		\label{F1score}
	\end{figure}
	
	The results give rise to the following observations.
	
	COP-Miner has higher accuracy than ARIMA and DTW prediction models. For example, on SDB1, the recall and F1-score of COP-Miner are 0.5950 and 0.7461, respectively, while those of ARIMA and DTW both are 0.4050 and 0.5765, respectively. This phenomenon can be found on most other datasets. The reason is as follows. The prediction results of ARIMA and DTW cannot accurately describe the trend characteristics, while COP-Miner is closer to the real data as a whole. Therefore, COP-Miner has better prediction accuracy than ARIMA and DTW.
	
	COP-Miner has higher accuracy than COP-noLEP. For example, on SDB4, the recall and F1-score of COP-Miner are 0.7742 and 0.8727, respectively, while those of COP-noLEP are 0.5556 and 0.7143. This phenomenon can be found on all other datasets. The reason for this is as follows. The difference between COP-Miner and COP-noLEP is that COP-Miner uses the LEP algorithm to extract keypoints, while COP-noLEP does not. The results indicate that LEP can improve accuracy by removing unimportant trend information in time series. Therefore, COP-Miner yields better prediction accuracy than COP-noLEP.
	
	\section{Conclusion} \label {section6}
	
	
	To discover all COPs with the same prefix pattern efficiently, we propose the COP-Miner algorithm, which consists of three parts: 1) extracting keypoints, 2) preparation stage, and 3) iteratively calculating supports and mining frequent COPs. To reduce distortion interference in time series and avoid mining redundant patterns, we propose the LEP algorithm to extract local extreme points and to obtain keypoint sub-time series and keypoint time series, which can effectively improve the prediction performance compared to mining directly on the original time series. To prepare for the first round of mining, we propose the method of the preparation stage, which contains four steps: obtaining the suffix OPP of the keypoint sub-time series, calculating the occurrences of the suffix OPP, verifying the occurrences of the keypoint sub-time series, and calculating the occurrences of all fusion patterns of keypoint sub-time series. To calculate the support effectively, we propose the CSS algorithm, which adopts a pattern fusion method to generate candidate patterns and an ending strategy to further prune redundant patterns. Experimental results from nine real datasets verified that COP-Miner yields better running performance and scalability than other ten competing algorithms, especially on large-scale datasets, COP-Miner is about 40 times and 4 times faster than OPP-Miner and EFO-Miner, respectively. More importantly, COPs with keypoint alignment yield better prediction performance.
 
 \section*{Acknowledgement}

This work was supported by Hebei Social Science Foundation Project (No. HB19GL055). 

\bibliographystyle{plain}
\bibliography{main}
%

\end{document}